\documentclass[journal]{IEEEtran}

\usepackage[latin1]{inputenc}
\usepackage{graphicx}
\usepackage{amssymb,amsmath,amsfonts,amsthm}
\usepackage{algorithm}
\usepackage{cite}
\usepackage{float}
\usepackage{graphics}
\usepackage{subfigure}
\usepackage{xspace}
\usepackage[usenames,dvipsnames]{color}
\usepackage{amssymb, epsfig}

\usepackage[noend]{algorithmic}

\usepackage{bbm}
\usepackage{subeqnarray}
\usepackage{amsmath}

\newcommand{\rank}{\operatorname{rank}}

\newcommand{\diag}{\operatorname{diag}}
\newtheorem{theorem}{Theorem}
 
\newtheorem{proposition}{Proposition}

\begin{document}

\title{Spatial Domain Simultaneous Information and Power Transfer for MIMO Channels}

\author{Stelios Timotheou, Member, IEEE, Ioannis Krikidis, Senior Member, IEEE,
Sotiris Karachontzitis, Student Member IEEE, and Kostas Berberidis, Senior Member, IEEE
\thanks{S. Timotheou and I. Krikidis are with KIOS Research Center for Intelligent Systems and Networks, University of Cyprus, Cyprus (E-mail: \text{\{timotheou.stelios, krikidis\}@ucy.ac.cy}).}
\thanks{S. Karachontzitis and K. Berberidis are with the Computer Engineering and Informatics Department, University of Patras, Greece (E-mails: \{karaxodg, berberid\}@ceid.upatras.gr) }
\thanks{Part of this work was supported by the Research Promotion Foundation, Cyprus under the project KOYLTOYRA/BP-NE/0613/04 ``Full-Duplex Radio: Modeling, Analysis and Design (FD-RD)''.}
\thanks{This work was supported in part by the EU and national funds via the National Strategic Reference Framework (NSRF)--Research Funding Program ``Thales'' (project ENDECON) and in part by the University of Patras, Greece.}
\thanks{Preliminary results of this work have been presented in \cite{greencomm13}.}
}
\maketitle

\begin{abstract}
In this paper, we theoretically investigate a new technique for simultaneous information and power transfer (SWIPT) in multiple-input multiple-output (MIMO) point-to-point with radio frequency energy harvesting capabilities.  The proposed technique exploits the spatial decomposition of the MIMO channel and uses the eigenchannels either to convey information or to transfer energy. In order to generalize our study, we consider channel estimation error in the decomposition process and the interference between the eigenchannels. An optimization problem that minimizes the total transmitted power subject to maximum power per eigenchannel, information and 
energy  constraints is formulated as a mixed-integer nonlinear program and solved to optimality using mixed-integer second-order cone programming. A near-optimal mixed-integer linear programming solution is also developed with robust computational performance. A polynomial complexity algorithm is further proposed for the optimal solution of the problem when no maximum power per eigenchannel constraints are imposed. In addition, a low polynomial complexity algorithm is developed for the power allocation problem with a given eigenchannel assignment, as well as a low-complexity heuristic for solving the eigenchannel assignment problem. 
\end{abstract}


\vspace{-0.4cm}
\section{Introduction}

\IEEEPARstart{T}{he} integration of renewable energy sources into communication networks is a hot research topic. It provides significant energy gains and is an efficient green communication solution for the expected future data traffic increase. Traditional renewable energy sources such as solar energy and wind depend on the weather conditions and are characterized by high instability; the integration of these energy sources requires a fundamental re-design of communications systems in all levels of protocol stack in order to ensure robustness and reliability \cite{SUD}.  Several works consider conventional renewable energy resources and study optimal resource allocation techniques for different objective functions and network topologies e.g., \cite{EPH,OZE,GUR,BLA,TUT}. However, the intermittent and unpredictable nature of these energy sources makes energy harvesting (EH) critical for applications where quality-of-service (QoS) is of paramount importance. On the other hand, recently there is a lot of interest to use electromagnetic radiation as a potential renewable energy resource. The key idea of this concept is that electromagnetic (EM) waves convey energy, which can be converted to DC-voltage by using specific rectenna circuits \cite{MON,VOL}. 

Most of the work on radio-frequency (RF) energy transfer concerns the design of rectenna circuits in different frequency bands, which is a fundamental aspect towards the development of this technology \cite{POP2,SUN,NIN}. From an information theoretic standpoint, evaluating the channel capacity for different network configurations,  when RF EH requirements characterize the receiver nodes, is a challenging problem. The work in \cite{VAR1} discusses the joint transfer of information and energy for a single-input single-output (SISO) channel and is extended in \cite{VAR2} for a set of parallel point-to-point channels; the authors in \cite{SIM} study the capacity for  two baseline multi-user systems with RF energy constraints, namely multiple access and multihop channels. However, information theoretic studies assume that the receivers are able to decode information and harvest energy from the same RF signal without limitations. Although this assumption provides some useful theoretical bounds, it cannot be supported by the current practical implementations.

In order to satisfy the above practical limitation, the work in \cite{RUI1} deals with the beamforming design for a basic broadcast multiple-input multiple-output (MIMO) channel, where the source conveys information to one receiver and transfers energy to the other one. In that work, the authors introduce two main practical techniques for simultaneous information and power transfer (SWIPT): a) ``time switching'' (TS), where the receiver switches between decoding information and harvesting energy and b) ``power splitting'' (PS), where the receiver splits the received RF signal in two parts for decoding information and harvesting energy, respectively. This work is extended in \cite{XIA} for imperfect channel information at the transmitter using robust optimization. 

The employment of the above two practical approaches in different fundamental network structures, is a hot research topic and several recent works appear in the literature. The works in \cite{RUI2,RUI3,KRI,CFL} focus on the TS technique for different network topologies. Specifically, in \cite{RUI2} the authors investigate the optimal switching strategy for a SISO channel in order to achieve various trade-offs between wireless information transfer and EH with/without channel state information (CSI) knowledge at the transmitter. The work in \cite{RUI3} presents a time division multiple access (TDMA)-based multiuser broadcast channel where the downlink (broadcast channel) is used for EH while the uplink for conveying information from the users to the access point. In \cite{KRI}, the authors apply the TS technique for a basic relay channel and investigate the optimal switching rule. A TS switching architecture is proposed in \cite{CFL} for a multiple-input single-output (MISO) system, where the single-antenna receiver uses its harvested energy to feed a quantized CSI.  On the other hand, in \cite{NAS} the authors study the performance of a cooperative system, where the relay is powered by employing a PS technique on the received signal. The work in \cite{TIM} studies the optimal transmitted power for a MISO interference channel with PS where the destinations have both information/energy constraints. The TS and PS techniques have also been exploited in cognitive radio networks to achieve joint information and energy cooperation between a primary and a secondary system, resulting in improved performance for both systems due to the added RF EH capabilities \cite{ganCR}. Recently, the authors in \cite{XUNZH} take into account the energy consumption of the active electronic circuits (e.g., mixers) and they propose an integrated receiver architecture, where PS is employed at the baseband domain. A recent overview of SWIPT systems with a particular focus on the hardware realization of rectenna circuits and practical techniques that achieve SWIPT in the domains of time, power and antennas is provided in \cite{magKrik}.

Although most of the previous studies refer to specific and fix network topologies, more recent works study TS/PS for large scale networks with location randomness. In \cite{LEE}, the authors study the interaction between primary and cognitive radio networks,  where cognitive radio nodes can harvest energy from the primary transmissions, by modeling node locations as Poisson point processes. A cooperative network with multiple transmitter-receiver pairs and a single EH relay is studied in \cite{DIN2} by taking the spatial randomness of user locations into consideration. In \cite{KRI5}, the authors study the performance of a large scale ad hoc network with/without relaying where receivers have PS capabilities.

In contrast to the  conventional TS and PS approaches, we propose a new technique for SWIPT in the spatial domain for a basic point-to-point MIMO channel. In this work, spatial domain does not refer to the antenna elements \cite{RUI2,KRI6} but mainly on the spatial degrees of freedom of the channel.  Based on the singular value decomposition (SVD) of the MIMO channel, the communication link is transformed to parallel  channels that can convey either information data or energy; this binary allocation is in respect to the current practical limitations. In order to make our analysis more general (and practical), we assume an imperfect channel estimation that affects the orthogonality of the eigenchannels (the imperfect channel knowledge generates an interference to the eigenchannels). We study the minimization of the transmitted power when the receiver is characterized by maximum power per eigenchannel, information rate and power transfer constraints which is a mixed-integer nonlinear optimization problem and propose several solution methods. It is worth noting that the main purpose of this work is to introduce a new technique for SWIPT in MIMO systems; this technique is studied from a theoretical standpoint and practical implementation is beyond the scope of this paper.

More specifically, the contributions of this paper are:
\begin{itemize}
\item Construction of mixed-integer second-order cone programing (MISOCP) and mixed-integer linear programming (MILP) formulations to the general problem that can be solved using standard optimization tools to provide optimal and near-optimal solutions, respectively.
\item Development of an optimal polynomial complexity algorithm for the solution of the special problem of having no maximum power per eigenchannel constraints. A low complexity polynomial algorithm is also developed for the general problem with excellent performance.
\item Derivation of a waterfilling-like procedure for the solution of the optimal power allocation problem with fixed eigenchannel assignment, of low polynomial complexity. 
\end{itemize}

\noindent \underline{Notation:} Upper case and lower case bold symbols denote matrices and vectors, respectively, while calligraphic letters denote sets.
$\diag({\mathbf{x}})$ represents a diagonal matrix with the vector $\mathbf{x}$ in the main diagonal, $\det(\cdot)$ denotes determinant, $\mathbf{I}_n$ is
the identity matrix of order $n$, $\log(\cdot)$ denotes the
logarithm of base $2$, $\mathbb{E}[\cdot]$ represents the
expectation operator, the superscripts $H$ and $*$ denote the Hermitian
transpose and conjugate operations, respectively. $\left|\mathcal{A}\right|$ denotes the cardinality of set $\mathcal{A}$, while $\mathcal{A} \setminus \mathcal{B}$ is the difference between sets $\mathcal{A}$ and $\mathcal{B}$.

The paper is organized as follows.  Section II presents the system model, introduces the spatial domain EH and the associated optimization problem. Section III deals with the optimal solution of (a) the general problem considered using MISOCP and MILP formulations and (b) the special case with no maximum power per eigenchannel. Section IV develops an optimal procedure for the power allocation problem for a given eigenchannel assignment. Section V introduces a low-complexity suboptimal heuristic for the considered optimization problem. Section VI discusses the numerical performance of the proposed schemes and Section VII concludes the paper. 

\section{System model \& problem formulation}

We assume a simple point-to-point MIMO model consisting of one source with $N_T$ transmit antennas and one destination with $N_R$ receive antennas. The source is connected to a constant power supply while the destination has RF transfer capabilities and can harvest energy from the received electromagnetic radiation.  We consider a flat fading spatially uncorrelated Rayleigh MIMO channel where $\mathbf{H}\in \mathbb{C}^{N_R\times N_T}$ denotes the channel matrix. The channel remains constant during one transmission time-slot and changes independently from one slot to the next. The entries of $\mathbf{H}$ are assumed to be independent, zero-mean
circularly symmetric complex Gaussian (ZMCSCG) random
variables with unit variance (which ensures $\rank(\mathbf{H}) = N=\min\{N_T,N_R\}$). The received signal is described by 
\begin{align}
\mathbf{y}=\mathbf{H}\mathbf{x}+\mathbf{n},
\end{align}

\noindent where $\mathbf{x}\in \mathbb{C}^{N_T\times 1}$ denotes the transmitted signal with $\mathbb{E}[\mathbf{x}\mathbf{x}^H]=\mathbf{Q}$ and $\mathbf{n}\in \mathbb{C}^{N_R\times 1}$ represents the noise vector having ZMCSCG entries of unit variance.  We assume that the channel matrix is subject to a channel estimation error and therefore is
{\it imperfectly} known at both the transmitter and the receiver with a MMSE estimation
error $\mathbf{E}\triangleq \mathbf{H}-\mathbf{\widehat{H}}$, where the entries of $\mathbf{E}$ are ZMCSCG with variance
$\sigma_{\epsilon}^2$ and the entries of
$\mathbf{\widehat{H}}$ are also independent and identically distributed (i.i.d.) ZMCSCG with variance
$1-\sigma_{\epsilon}^2$ \cite{YOO}.  The channel estimation is performed at the destination via a downlink pilot sequence and is communicated to the source by using an uplink feedback channel; the channel estimation process is beyond the scope of this paper. We note that the parameter $\sigma_{\epsilon}^2$ captures the quality of the channel estimation and is assumed to be known to both the transmitter and the receiver. The work in \cite{YOO} provides some examples for $\sigma_{\epsilon}^2$ related to different channel estimation schemes and channel statistics; in this paper, we assume that $\sigma_{\epsilon}^2$ is a constant without further implications. Based on the estimated channel $\mathbf{\widehat{H}}$, a lower bound of the instantaneous mutual information is given by \cite{YOO}
\begin{align}
I(\mathbf{x};\mathbf{y})=\log \det \left(\mathbf{I}_{N_R}+\frac{1}{1+\sigma_{\epsilon}^2 P}\mathbf{\widehat{H}}\mathbf{Q}\mathbf{\widehat{H}}^H \right).
\end{align}

\noindent It is worth noting that for $\sigma_{\epsilon}^2=0$ (perfect channel knowledge), the above expression gives the exact mutual information of the MIMO channel. By using the SVD of the $\mathbf{\widehat{H}}$ channel, it has been proven in \cite{YOO} that the lower bound of the mutual information is maximized when $\mathbf{Q}=\diag(P_1,\ldots,P_{N_T})$ and takes the form 
\begin{align}
I(\mathbf{x};\mathbf{y})=\sum_{i\in\mathcal{N}}\log \left(1+\frac{P_i \lambda_i}{1+\sigma_{\epsilon}^2P}  \right), \label{rate}
\end{align}

\noindent where $\lambda_i$, $i\in\mathcal{N}=\{1,...,N\}$, is the $i$-th eigenvalue of $\mathbf{\widehat{H}}\mathbf{\widehat{H}}^H$, $P_i$ is the power allocated to the $i$-th eigenchannel and $P=\sum_{i}P_i$. It is further assumed that $\lambda_1\geq \lambda_2\geq ...\geq\lambda_N$. 

\begin{figure}[t]
\centering
\includegraphics[keepaspectratio,width=\columnwidth]{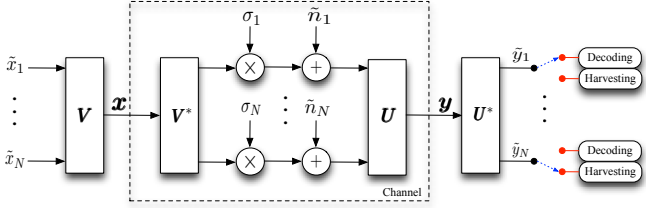}
\vspace{-0.5cm}
\caption{The SVD transformation of the MIMO channel into $N$ parallel AWGN channels.}\label{model_SVD}
\end{figure}

This bound can be achieved using SVD on the estimated matrix $\mathbf{\widehat{H}}$ and by appropriate precoding and receiver shaping at the source and the destination, respectively. Fig. \ref{model_SVD} schematically presents the system model and the transformation of the MIMO channel to $N$ eigenchannels for potential transfer of information/energy. We show the SVD coordinate transformations when the matrix $\widehat{\mathbf{H}}$ has a SVD equal to $\widehat{\mathbf{H}}=\mathbf{U}\mathbf{\Sigma}\mathbf{V}^*$ , where $\mathbf{U}\in \mathbb{C}^{N_R\times N_R}$ and $\mathbf{V}\in \mathbb{C}^{N_T\times N_T}$ are unitary matrices and $\mathbf{\Sigma}\in\mathbb{C}^{N_R\times N_T}$ is a diagonal matrix with elements the singular values of the matrix $\widehat{\mathbf{H}}$ i.e., $\sigma_i$, $i\in\mathcal{N}$ and $\sigma_i=\sqrt{\lambda_i}$. Based on the SVD, the channel input $\mathbf{x}$ comes from the transformation of the original sequence $\tilde{\mathbf{x}}$ with the matrix $\mathbf{V}$ while the received sequence $\mathbf{y}$ is transformed to the sequence $\tilde{\mathbf{y}}$ by using the matrix $\mathbf{U}$. In the ideal case where $\sigma_{\epsilon}^2=0$, the MIMO channel is decomposed to $N$ parallel  SISO channels with 
\begin{align}
\tilde{y}_i=\sigma_i \tilde{x}_i+\tilde{n}_i,\;\;\;\text{with}\;\;i\in\mathcal{N},
\end{align}     
where $\tilde{n}_i$ denotes the AWGN for the $i$-th parallel channel and has the same distribution with $n_i$ (due to the unitary transformation) in reception.  

The expression in \eqref{rate} shows that the imperfect channel estimation generates an interference to the $N$ parallel SISO channels (eigenchannels). We assume that the destination is characterized by both information rate and RF EH requirements; this means that for the duration of one transmission  the destination requires an information rate $C_I$ and energy $\bar{C}_{EH}=C_{EH}/\delta$ as input to its rectenna circuits (the amount of energy that can be stored, $C_{EH}$, depends on the energy conversion efficiency of the specific implementation, $\delta$) e.g., \cite{TIM,SHI}. For simplicity and without loss of generality we assume in the mathematical analysis that  $\delta=1$.

\subsection{SWIPT optimization problem}

The proposed scheme exploits the SVD structure of the MIMO channel and achieves SWIPT in the spatial domain. More specifically, the transformation of the MIMO channel to $N$ parallel SISO channels allows the simultaneous transfer of data traffic and RF energy by using an eigenchannel either to convey information or energy. An eigenchannel cannot be used to convey both information and energy; this limitation refers to practical constraints and is inline with the other approaches proposed in the literature, e.g. power splitting. Based on the diagram in Fig. \ref{model_SVD}, at the output of each eigenchannel there is a switcher which drives the channel output either to the conventional decoding circuit or to the rectification circuit. During each transmission an appropriate optimization problem is solved, which determines the usage of each antenna and a switching mechanism selects the appropriate circuits.

In this paper, we focus on the minimization of the transmitted power given that both information and energy constraints are satisfied.  Based on the notation considered, the proposed technique introduces the following optimization problem
\begin{subequations}
\label{eq:ProblemGeneral}
\begin{align}
\min & ~P=\sum_{i\in\mathcal{N}} P_i \label{eq:ProblemGenerala}\\
\text{s.t.}\;\;&\sum_{i\in\mathcal{N}} \log \left(1+\frac{\pi_i P_i \lambda_i}{1+\sigma_{\epsilon}^2 P}  \right)\geq C_I, \label{eq:ProblemGeneralb}\\
&\sum_{i\in\mathcal{N}} (1-\pi_i)\big(P_i \lambda_i+\sigma_{\epsilon}^2P \big)\geq C_{EH}, \label{eq:ProblemGeneralc} \\
&0 \leq P_i \leq P_{max},~\pi_i \in \{0,1 \}, ~i\in\mathcal{N},\label{eq:ProblemGenerale}
\end{align}
\end{subequations}
where $P_{max}$ indicates the maximum power that can be used in each eigenchannel, while binary variable $\pi_i$ indicates whether the $i$-th eigenchannel is used for information transfer ($\pi_i=1$) or energy transfer ($\pi_i=0$). Note that terms $\log \left(1+\frac{\pi_i P_i \lambda_i}{1+\sigma_{\epsilon}^2 P}\right)$ are equal to the more intuitive representation $\pi_i \log \left(1+\frac{P_i \lambda_i}{1+\sigma_{\epsilon}^2 P} \right)$ when $\pi_i\in\{0,1\}$, $i\in\mathcal{N}$. This mathematical program involves binary and continuous variables, as well as nonlinear functions; hence it belongs to the class of mixed-integer nonlinear optimization problems, which are very hard to solve in general.

The incorporation of $\sigma_{\epsilon}^2$ into \eqref{eq:ProblemGeneral} makes the problem more realistic and general, with interesting implication effects. On the one hand, increasing $\sigma_{\epsilon}^2$ reduces the achieved information rate in \eqref{eq:ProblemGeneralb}. On the other hand, the added cross-eigenchannel interference becomes a useful source of energy transfer in \eqref{eq:ProblemGeneralc}, which implies that even zero-power EH eigenchannels can contribute to the satisfaction of the EH constraint. From an optimization point of view, the problem becomes more difficult to tackle due to the presence of term $\sigma_{\epsilon}^2P$ in the denominator of the logarithmic terms in \eqref{eq:ProblemGeneralb}. Hence, the developed solution methodologies are quite general and can be applied not only for the solution of the special case with $\sigma_{\epsilon}^2=0$, but also for other problems with similar information rate expressions.

We note that the energy harvested due to the receiver noise is negligible. 
In addition, the optimization problem can be solved at the source node and then control bits can be used in order to inform the receiver about the content of each eigenchannel. 

\subsection{Implementation issues}

The main challenge to implement the proposed SWIPT technique is to perform the required signal processing, including the eigenchannel decomposition, with ultra-low power consumption. This requires the use of analog passive electronic elements (e.g., Schottky diode), which is in line with current implementations of conventional wireless power transfer \cite{borges14,kim14}. In respect to these requirements, analog eigenmode beamforming with passive electronic elements (180-degree hybrid couplers) is a promising technology recently proposed  \cite{murata}, that could be used for the practical implementation of the proposed technique. This implementation technology achieves diagonalization of the MIMO channel matrix without destroying the energy content of the received signal.

It is worth noting that although digital beamforming dominates current communication systems, analog beamforming is an emerging technology with particular impact on future millimeter-wave (mmWave)- based systems \cite{roh,han}. In these systems, the transceivers are composed of a large number of antennas and the conventional baseband/digital processing requires a large number of analog-to-digital converters and RF chains, which results in a high complexity and power consumption. Also note that SWIPT is only studied theoretically and not used in current devices; proof-of-concept implementations are still under investigation.

\vspace{-0.3cm}
\section{Optimal Eigenchannel Assignment and Power Allocation}

The considered problem is nonlinear and combinatorial in nature, and hence very hard to solve. Moreover, standard MILP/MISOCP solvers are not directly applicable for its solution due to the presence of nonlinear terms. In this section, two formulations are proposed for the optimal/near-optimal solution of problem \eqref{eq:ProblemGeneral} using standard solvers. The first, transforms the problem into an equivalent MISOCP formulation that provides the optimal solution, while the second approximates the logarithmic functions using piecewise linear approximation (PLA) to arrive at a near-optimal (in most cases optimal) result by solving a series of MILP problems. These approaches have exponential complexity and are only used for benchmarking purposes; the reason for discussing both is that the MISOCP guarantees optimality, while the MILP has more robust performance (see Section \ref{sec:nume}). In addition, we provide a polynomial algorithm that optimally solves the problem for the special case of $P_{max}=\infty$.

\vspace{-0.3cm}
\subsection{Optimal MISOCP solution}
\label{sec:MISOCP}

Towards the optimal solution of \eqref{eq:ProblemGeneral}, we write it as:
\begin{subequations}
\label{eq:SOCP1}
\begin{align}
\min & ~P = \sum_{i\in\mathcal{N}} P_{i}\label{eq:SOCP1a}\\
\text{s.t.} & \prod_{i\in\mathcal{N}}\left(\frac{1+\sigma_{\epsilon}^{2}P+\pi_{i}P_{i}\lambda_{i}}{1+\sigma_{\epsilon}^{2}P}\right)\ge2^{C_{I}}, \label{eq:SOCP1b}\\
 & \sum_{i\in\mathcal{N}} (P_{i} - \pi_{i}P_{i})\lambda_{i} + N\sigma_{\epsilon}^{2}P-\sigma_{\epsilon}^{2}\sum_{i\in\mathcal{N}}\pi_{i}P\ge C_{EH}, \label{eq:SOCP1c}\\
 & 0\le P_{i}\leq P_{max},\,\,\,\pi_{i}\in\left\{ 0,1\right\},~i\in\mathcal{N}. \label{eq:SOCP1e}
\end{align}
\end{subequations}

Formulation \eqref{eq:SOCP1} is not MISOCP due to the products of variables, $\pi_{i}P_{i}$
and $\pi_{i}P$, $i\in\mathcal{N}$, in (\ref{eq:SOCP1b})-(\ref{eq:SOCP1c})
and the general product of terms in (\ref{eq:SOCP1b}). Next we
show how these can be transformed into equivalent representations of MISOCP form. Firstly, we state two auxiliary propositions.
\begin{proposition}
\label{prop:prod}
Let $w\in\left\{ 0,1\right\} $ and $x\in[x^{LB},\, x^{UB}]$.
Constraint $y=wx$ can be equivalently represented with the following linear constraints:
\begin{subequations}
\label{eq:prod}
\begin{align}
wx^{LB} \leq y & \leq x-(1-w)x^{LB}, \label{eq:proda}\\
x-(1-w)x^{UB} \leq y & \leq wx^{UB}, \label{eq:prodb}
\end{align}
\end{subequations}
\end{proposition}
\begin{proof}
The proof is easily derived by checking that $w=0$ and $w=1$ yield $y=0$ and $y=x$, respectively.
\end{proof}
Let us define $P_{i}^{\pi}=\pi_{i}P_{i}$ and $P_{i}^{tot}=\pi_{i}P$, $i\in\mathcal{N}$,  which appear in constraints \eqref{eq:SOCP1b} and \eqref{eq:SOCP1c}.
Based on proposition \ref{prop:prod}, equalities $P_{i}^{\pi}=\pi_{i}P_{i},\,i\in\mathcal{N}$
can be represented by constraints (\ref{eq:proda})-(\ref{eq:prodb})
with $y\equiv P_{i}^{\pi}$, $w \equiv \pi_{i}$, $x \equiv P_{i}\in[0,\, P_{max}]$; similarly we can represent
$P_{i}^{tot}=\pi_{i}P$, with $y\equiv P_{i}^{tot}$ $w \equiv \pi_{i}$, $x \equiv P\in[0,\, NP_{max}]$. Next, we will show that the general product appearing in (\ref{eq:SOCP1b}) can be represented by a hierarchy of convex second-order rotated cone
(SORC) constraints of the form $x_{1}x_{2}\geq x_{0}^{2},\, x_{1}\ge0,\, x_{2}\ge0$
using proposition \ref{prop:SOCPdecomp}. 
\begin{proposition}{(\cite{book_nemirovski}, p.105)}: 
\label{prop:SOCPdecomp}
The geometric mean constraint (GMC) $x_{1}...x_{2^{l}}\geq t^{2^{l}},\, x_{i}\ge0,\, i=1,...,2^{l}$
is convex and can be represented by a hierarchy of SORC constraints.
Let us define existing variables as $x_{i}\equiv x_{0,i}$ and new
variables $x_{k,i}\ge0,\, k=1,...,l,\, i=1,...,2^{l-k}$. Then, the
GMC is equivalent to the following SORC constraints
\begin{align*}
\text{layer \ensuremath{k}: } & x_{k-1,2i-1}x_{k-1,2i} \geq x_{k,i}^{2},i=1,...,2^{l-k},k=1,..., l,\\
 & x_{l,1}\geq t.
\end{align*}
\end{proposition}
The proof is based on the fact that each of these constraints is SORC,
hence convex, while if these constraints hold the GMC holds as well. To see why this is true, let us consider the case $x_{1}x_{2}x_{3}x_{4}\geq t^{4}$.
Setting $x_{i}\equiv x_{0,i},\, i=1,..,4$, this constraint becomes:
\begin{subequations}
\label{eq:exampleGM}
\begin{align}
x_{0,1}x_{0,2}  \geq & x_{1,1}^{2},~x_{0,3}x_{0,4}  \geq x_{1,2}^{2},\label{eq:exampleGMa}\\
x_{1,1}x_{1,2}  \geq & x_{2,1}^{2},~x_{2,1} \geq  t. \label{eq:exampleGMd}
\end{align}
\end{subequations}
The two constraints in \eqref{eq:exampleGMa} yield that $x_{0,1}x_{0,2}x_{0,3}x_{0,4} \geq x_{1,1}^{2}x_{1,2}^{2}$, which combined with \eqref{eq:exampleGMd} gives $x_{1,1}^2 x_{1,2}^2  \geq x_{2,1}^{4}\geq t^4$; hence, if \eqref{eq:exampleGMa}-\eqref{eq:exampleGMd} hold, the original constraint also holds.

To bring (\ref{eq:SOCP1b}) into the GMC form, we define $M=\left\lceil \log N\right\rceil $ and $N_{u}=2^{M}$,
and rewrite it as:
\small
\[
\prod_{i=1}^{N}\left(1+\sigma_{\epsilon}^{2}P+P_{i}^{\pi}\lambda_{i}\right)\prod_{i=1}^{N_{u}-N}\left(1+\sigma_{\epsilon}^{2}P\right)\geq(2^{C_{I}/N_{u}}(1+\sigma_{\epsilon}^{2}P))^{N_{u}}.
\]
\normalsize
Following proposition \ref{prop:SOCPdecomp}, we can represent (\ref{eq:SOCP1b}) with the hierarchy of SORC constraints by setting $x_{0,i}=1+\sigma_{\epsilon}^{2}P+P_{i}^{\pi}\lambda_{i}\ge0,$ $i\in\mathcal{N}$, $x_{0,i}=1+\sigma_{\epsilon}^{2}P\ge0,$ $i=N+1,...,N_{u}$ and $t=2^{C_{I}/N_{u}}(1+\sigma_{\epsilon}^{2}P)$.
Towards this direction, a total of $N_{u}-1$ new variables are needed,
$x_{k,i}\ge0,k=1,...,M,\, i=1,...,2^{M-k}$, and $N_{u}/2$ SORC constraints.

Based on the above analysis, the original problem can be rewritten
as:
\small
\begin{subequations}
\label{eq:SOCP2}
\begin{align}
\min & ~P=\sum_{i\in\mathcal{N}} P_{i}\label{eq:SOCP2a}\\
\text{s.t. } & x_{k-1,2i-1}x_{k-1,2i}\geq x_{k,i}^{2},~k=1,...,M,i=1,...,2^{M-k},\label{eq:SOCP2b}\\
 & x_{M,1}\geq t,~ t=2^{C_{I}/N_{u}}(1+\sigma_{\epsilon}^{2}P),\label{eq:SOCP2c}\\
 & x_{0,i}=1+\sigma_{\epsilon}^{2}P + P_i^{\pi}\lambda_i,~i\in\mathcal{N},\label{eq:SOCP2e}\\
 & x_{0,i}=1+\sigma_{\epsilon}^{2}P,~i=N+1,...,N_{u},\label{eq:SOCP2f}\\
 & x_{k,i}\geq 0, k=0,...,M, i=1,...,2^{M-k},\label{eq:SOCP2g}\\
 & \sum_{i\in\mathcal{N}} (P_{i}-P_{i}^{\pi})\lambda_{i}+\sigma_{\epsilon}^{2}(NP-\sum_{i\in\mathcal{N}} P_{i}^{tot})\geq C_{EH},\label{eq:SOCP2h}\\
 & P_{i}^{\pi}\leq\pi_{i}P_{max},\, P_{i}^{\pi}\leq P_{i},~i\in\mathcal{N}, \label{eq:SOCP2j}\\
 & P_{i}^{\pi}\geq P_{i}-(1-\pi_{i})P_{max},\, P_{i}^{\pi}\ge0,~i\in\mathcal{N}, \label{eq:SOCP2k}\\
 & P_{i}^{tot}\leq\pi_{i}NP_{max},\, P_{i}^{tot}\leq P,~i\in\mathcal{N}, \label{eq:SOCP2l}\\
 & P_{i}^{tot}\geq P-(1-\pi_{i})NP_{max},\, P_{i}^{tot}\ge0,~ i\in\mathcal{N}, \label{eq:SOCP2m}\\
 & 0\le P_{i}\leq P_{max},\,\pi_{i}\in\left\{ 0,1\right\},~ i\in\mathcal{N}. \label{eq:SOCP2n}
\end{align}
\end{subequations}
 \normalsize

In formulation \eqref{eq:SOCP2}, (\ref{eq:SOCP2b})-(\ref{eq:SOCP2g}) represent
the information constraint, (\ref{eq:SOCP2h}) the EH constraint,
while (\ref{eq:SOCP2j}),
(\ref{eq:SOCP2k}) and (\ref{eq:SOCP2l}), (\ref{eq:SOCP2m}) are
used to linearize equations $P_{i}^{\pi}=\pi_{i}P_{i}$ and $P_{i}^{tot}=\pi_{i}P$,
respectively. Note that formulation \eqref{eq:SOCP2} is MISOCP and can be solved with
standard optimization solvers.

\subsection{Near-optimal MILP solution}
\label{sec:MILP}

In this section we propose an iterative approach for the solution of \eqref{eq:ProblemGeneral}, where in each iteration a MILP problem is solved.
Although, this approach yields near-optimal results (in most cases optimal) it is considerably faster than MISOCP for large problems. Towards this direction, we examine the implications of having a fixed value for the total power $P$ in \eqref{eq:ProblemGeneralb}.  Assuming that this value is equal to $P_f$, the problem becomes:
\begin{subequations}
\label{eq:MILP1}
\begin{align}
\min & ~P = \sum_{i\in\mathcal{N}} P_i \label{eq:MILP1a}\\
\text{s.t.}\;\;&\sum_{i\in\mathcal{N}} \log \left(1+\frac{\pi_i P_i \lambda_i}{1+\sigma_{\epsilon}^2 P_f}  \right)\geq C_I, \label{eq:MILP1b}\\
&\sum_{i\in\mathcal{N}} (1-\pi_i)\left(P_i \lambda_i+\sigma_{\epsilon}^2P \right)\geq C_{EH}, \label{eq:MILP1c} \\
&0 \leq P_i \leq P_{max},~\pi_i \in \{0,1 \},~i\in\mathcal{N}.\label{eq:MILP1e}
\end{align}
\end{subequations}
Using the transformations, $P_{i}^{\pi}=\pi_{i}P_{i}$ and $P_{i}^{tot}=\pi_{i}P$,  introduced in section \ref{sec:MISOCP} via proposition 1, all the constraints of problem \eqref{eq:MILP1} become of MILP form, except from constraint \eqref{eq:MILP1b}, due to the presence of terms $f_i(P_i^{\pi}) = \log\left(1+\frac{P_i^{\pi} \lambda_i}{1+\sigma_{\epsilon}^2 P_f}\right)$. Nonetheless, these terms can be transformed into MILP form by considering appropriate piecewise linear approximation (PLA) functions $g_i(P_i^{\pi}) \approx f_i(P_i^{\pi})$. PLA functions $g_i(P_i^{\pi})$ can also be selected to provide either upper or lower bounds (UB or LB) to the optimal value of problem \eqref{eq:MILP1}, if $g_i(P_i^{\pi}) \leq f_i(P_i^{\pi})$ or $g_i(P_i^{\pi}) \geq f_i(P_i^{\pi})$, $0\le P_i^{\pi} \leq P_{max}$, $i\in\mathcal{N}$, respectively.  Constraints $g_i(P_i^{\pi}) \leq f_i(P_i^{\pi})$ and $g_i(P_i^{\pi}) \geq f_i(P_i^{\pi})$ can be achieved by taking the lower envelop of a number of lines that intersect $f_i(P_i^{\pi})$ from below or are tangent to $f_i(P_i^{\pi})$ from above respectively, as illustrated in Fig. \ref{fig:PLA}.

\begin{figure}[t]
\centering
\includegraphics[width=\linewidth]{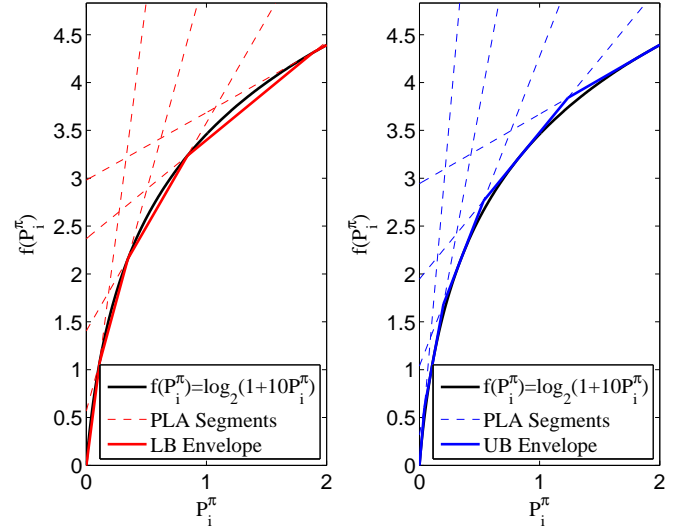}
\caption{Piecewise linear approximation; $P_i^{max}=2$; and $\lambda_i/(1+\sigma_{\epsilon}^{2}P_f)=10$.}
\label{fig:PLA}
\end{figure}

One can easily find the PLA of a logarithmic function, by considering a fixed number of points, at integer multiples of a parameter $\kappa$, such that segments are either tangent (for UB) or intersect (for LB) the logarithmic function. Nonetheless, this approach has no control over the approximation error, while the maximum error for different segments varies. For this reason, an approach has been followed that produces segments when needed so that the approximation error does not exceed a predefined value $e_{max}$ \cite{timotheouAssetTask}. Although this approach may lead to a large number of segments we can adjust $e_{max}$ to achieve the desirable number of segments and at the same time maintain the error less than a constant known value.

The approximation procedure, involves two steps for the production of each segment. In LB-PLA, where an LB envelope needs to be produced, the first step for the derivation of segment $l$, starts from a known intersection point $\{P_{i,l}^{\pi}, f(P_{i,l}^{\pi})\}$ (initially, $\{P_{i,1}^{\pi}=0, f(P_{i,1}^{\pi})=0\}$) and finds a line $\alpha_{i,l} P_i^{\pi}+\beta_{i,l}$ such that $f(P_i^{\pi}) - (\alpha_{i,l} P_i^{\pi}+\beta_{i,l})\leq e_{max}$, for $f(P_i^{\pi}) \geq \alpha_{i,l} P_i^{\pi} + \beta_{i,l}$. In the second step, a new intersection point is found, such that $f(P_{i,l+1}^{\pi}) = \alpha_{i,l} P_{i,l+1}^{\pi} + \beta_{i,l}$, which completes the creation of the $l$th segment. In UB-PLA, where a UB envelope needs to be produced,  the first step for the derivation of segment $l$, starts from a point of maximum error, $\{P_{i,l}^{\pi}, f(P_{i,l}^{\pi}) + e_{max}\}$ (initially, the procedure starts from $\{P_{i,1}^{\pi}=0, f(P_{i,1}^{\pi})=0\}$ to produce such a point) and computes a line $\alpha_{i,l} P_i^{\pi}+b_{i,l}$ that is tangent to $f(P_i^{\pi})$. In the second step, a new point of maximum error is found such that $\alpha_{i,l} P_{i,l+1}^{\pi} + \beta_{i,l} = f(P_{i,l+1}^{\pi})+e_{max}$, which completes the creation of the $l$th segment. This two-step iterative procedure is followed until $P_{max}$ is reached.

Having obtained PLA functions $g_i(P_i^{\pi})$, $i\in\mathcal{N}$, we now describe how to obtain a problem formulation
with MILP constraints. Let us assume that function $f_i(P_i^{\pi})$ is approximated by $L_i$ linear segments with slopes
$\alpha_{i,1}$,...,$\alpha_{i,L_i}$ and start-points $P_{i,1}^{\pi}$,...,$P_{i,L_i}^{\pi}$. Let us also assume that
$P_{i,L_i+1}^{\pi}=P_{max}$. Because $f_i(P_i^{\pi})$ are monotonically increasing and concave, the PLA functions $g_i(P_i^{\pi})$, will also be concave, with segments whose slopes have monotone decreasing values:
$\alpha_{i,1} > \alpha_{i,2}$$>...>$$\alpha_{i,L_i}$. Let $\xi_{i,l}, l=1,...,L_i$ be the value of $P_i^{\pi}$ corresponding to
the $l$th linear segment so that $0 \leq \xi_{i,l} \leq P_{i,l+1}^{\pi} - P_{i,l}^{\pi},~l=1,...,L_i$. Under the assumption that
$\xi_{i,m} = P_{i,m+1}^{\pi} - P_{i,m}^{\pi} ,m=1,...,l-1$ when $\xi_{i,l}>0$, it is true that $P_i^{\pi} = \sum_{l=1}^{L_i}
\xi_{i,l}$ and also that $g_i(P_i^{\pi}) = \sum_{l=1}^{L_i} \alpha_{i,l} \xi_{i,l}$. In other words, $P_i^{\pi}$ can be replaced
by the sum of variables $\xi_{i,l}, l=1,...,L_i$ if we can ensure that the solution of the optimization problem will always
be such that each $\xi_{i,l}$ is nonzero only when the variables $\xi_{i,m},~m=1,\ldots,l-1$ have obtained their maximum value. 

To explain why this condition holds for problem \eqref{eq:MILP1}, note that variables $P_i^{\pi}$ should be as small as possible to yield minimum total power, while the impact of $g_i(P_i^{\pi})$, and hence of variables  $\xi_{i,l}$,  should be as large as possible to ensure satisfaction of constraint \eqref{eq:MILP1b} with minimum power. Because the slopes of the PLA segments have monotonically decreasing values, the first segment has the largest impact; this implies that $\xi_{i,1}$ will be the first variable to be assigned a nonzero value and only if $\xi_{i,1}$ becomes equal to its maximum value, variable
$\xi_{i,2}$ will be assigned a nonzero value. This ensures the assumption made that $\xi_{i,l}$ will be non-zero only if $\xi_{i,m},m=1,...,l-1$ have attained their maximum value. Based on the above discussion formulation \eqref{eq:MILP1} becomes:
\begin{subequations}
\label{eq:MILP2}
\begin{align}
\min & ~P=\sum_{i\in\mathcal{N}} P_i \label{eq:MILP2a}\\
\text{s.t.}&\sum_{i\in\mathcal{N}} \sum_{l=1}^{L_i} \alpha_{i,l} \xi_{i,l} \geq C_I, \label{eq:MILP2b}\\
 & P_i^{\pi} = \sum_{l=1}^{L_i} \xi_{i,l},~i\in\mathcal{N}, \label{eq:MILP2c}\\
 & 0 \leq \xi_{i,l} \leq P_{i,l+1}^{\pi} - P_{i,l}^{\pi},~i\in\mathcal{N},~l=1,...,L_i,\label{eq:MILP2d}\\
 & \text{Constraints \eqref{eq:SOCP2h} - \eqref{eq:SOCP2n}.}
\end{align}
\end{subequations}
Note that formulation \eqref{eq:MILP2} is MILP, as constraints \eqref{eq:MILP2a} - \eqref{eq:MILP2d}, which represent constraint \eqref{eq:MILP1b}, also become linear under the considered PLA scheme. Also, by considering UB-PLA (LB-PLA) for $f_i(P_i^{\pi})$ a lower-bound (upper bound) to problem \eqref{eq:MILP1} is obtained, as we need to use more (less) power to satisfy the information constraint. 

If the solution of formulation \eqref{eq:MILP2} is feasible and the total required power is less than $P_f$, i.e.  $P\le P_f$, then $P_f$ is an upper bound to the optimal total power of problem \eqref{eq:ProblemGeneral}, otherwise $P_f$ is a lower bound. Based on this observation, we can construct a bisection procedure to obtain a near-optimal value for the total power, according to algorithm \ref{alg:MILP}.

\begin{algorithm}[t]
	\caption{\textbf{: MILP solution to problem \eqref{eq:ProblemGeneral}}} 
	\begin{algorithmic} [1] 
 	\STATE  \textbf{Init.} $P_{LB} = 0$; $P_{UB} = NP_{max}$; $\epsilon = 10^{-6}$; $e_{max} = 10^{-3}$; $f=0$.
  \WHILE{$|P_{UB} - P_{LB}|\ge\epsilon$}
  		\STATE Set $P_f = (P_{UB} + P_{LB})/2$.
  		\STATE Compute $\alpha_{i,l}$ and $P_{i,l}^{\pi}$, $i\in\mathcal{N}$, $l=1,...,L_i$ according to $e_{max}$ using PLA.
  		\STATE Solve problem \eqref{eq:MILP2} using a MILP solver to obtain $P_i$, $\pi_i$, $i\in\mathcal{N}$.
  		\IF{(problem feasible) AND ($P \leq P_f$)}
				\STATE Set $P_{UB} = P_f$; $\pi^{*}_{i}=\pi_{i}$; $f=1$.
			\ELSE  		
				\STATE Set $P_{LB} = P_f$.
  		\ENDIF
	\ENDWHILE
	\IF{$f=1$}
			\STATE Solve the power allocation problem with fixed assignment $\pi^*_i$, $i\in \mathcal{N}$ (Section \ref{sec:powerAllocation}).
	\ELSE
  		\STATE Problem \eqref{eq:ProblemGeneral} is deemed infeasible.
  \ENDIF
\end{algorithmic}  
\label{alg:MILP}
\end{algorithm}

Loss of optimality in algorithm \ref{alg:MILP} is primarily due to the PLA of the logarithmic terms in the total information rate expression which may result in a non-optimal eigenchannel assignment. Note that global optimality is preserved when the optimal eigenchannel assignment is found, by optimally solving the power allocation problem with fixed eigenchannel assignment (line 11).

\subsection{Special Case: $P_{max}=\infty$}
\label{sec:optimalPinf}

In this section, we deal with the solution of problem \eqref{eq:ProblemGeneral} when $P_{max}=\infty$. We show that we can obtain the optimal assignment by examining a polynomial number of information and EH eigenchannel combinations and develop an algorithm that solves this problem optimally in $O(N^2)$. We also derive analytical results for obtaining the optimal power allocation for a fixed eigenchannel assignment, as well as for deciding whether a candidate assignment is potentially optimal prior to computing the power allocation. 

Assume that a given assignment of channels is made such that $\pi_{i_1}=1$, $i_1\in \cal{I}$ and $\pi_{i_2}=0$, ${i_2}\in \cal{E}$, where $\mathcal{I}$ and $\mathcal{E}$ denote the sets of information and EH eigenchannels, respectively. Due to this assignment, problem \eqref{eq:ProblemGeneral}, with  $P_{max}=\infty$, can be written as:
\begin{subequations}
\label{eq:ProblemFA_InfP}
\begin{align}
\min & ~P={\displaystyle \sum_{i\in\mathcal{N}} P_{i}} \label{eq:ProblemFA_InfPa}\\
\text{s.t.} & {\displaystyle \sum_{i_1\in{\cal I}}\log\left(1+\frac{P_{i_1}\lambda_{i_1}}{1+\sigma_{\epsilon}^{2}P}\right)\geq C_{I}}, \label{eq:ProblemFA_InfPb} \\
 & {\displaystyle \sum_{i_2\in{\cal E}}\left(P_{i_2}\lambda_{i_2}+\sigma_{\epsilon}^{2}P\right)\geq C_{EH}}, \label{eq:ProblemFA_InfPc}\\
 & P_{i}\ge0,\, i\in\mathcal{N}. \label{eq:ProblemFA_InfPd}
\end{align}
\end{subequations}

Three important remarks can be made regarding \eqref{eq:ProblemFA_InfP}.

\noindent \textbf{Remark 1.} At most one EH eigenchannel can have nonzero power.  

Remark 1 is true because there is no imposed upper bound on the transmitted power in
each eigenchannel. As a result, if a set of eigenchannels ${\cal E}$ are assigned to
satisfy the EH constraint, only the one with the largest
eigenvalue in this set, $\lambda_{e}$, $e\in{\cal E}$ will be nonzero. Despite the fact that the power assigned to other channels in ${\cal E}$ is zero, they still contribute
towards the satisfaction of the EH constraint due to the presence
of the term $\sigma_{\epsilon}^{2}P$ in the EH constraint.

\noindent  \textbf{Remark 2.}  Eigenchannels with $P_i=0$ are those with the overall smallest eigenvalues. 

This is because ``better'' eigenchannels are beneficial for both the satisfaction of the information and EH
constraints. Hence, if we consider two eigenchannel with $\lambda_i>\lambda_j$, $P_i=0$ and $P_j=c>0$, by simply switching the power of the two eigenchannel, i.e., $P_j=0$ and $P_i=c>0$, more benefit can be obtaining for the corresponding constraint without affecting the other constraint or the total power in any way.

\noindent  \textbf{Remark 3.}  At the optimal solution, no information eigenchannel should have zero power. 

To see why this is true, notice that eigenchannels with zero power can contribute to the EH constraint but not to the information constraint.  This is a result of the cross eigenchannel interference, $\sigma_{\epsilon}^2 P$, that is present in all eigenchannels due to the channel estimation error. 
Hence, zero power information eigenchannels should always be used as EH eigenchannels to have a positive contribution towards the satisfaction of the EH constraint. 

These remarks imply that in order to find the optimal assignment we need two different indices. The first index $e$ indicates the energy eigenchannel with a possibly non-zero value, while the second indicates the channel with the largest eigenvalue such that $P_{i_2}=0$, $i_2 \in \mathcal{E} \setminus \{e\}$. Based on the above observations, we can conclude that in order to find the optimal assignment, we only need to examine $O(N^2)$ different assignment combinations. The next result, summarizes the solution of each assignment combination. 

\begin{theorem}
Let $\mathcal{I},~\mathcal{E}$ be a fixed assignment, $e=\min{\{\mathcal{E}\}}$, $\zeta = \sum_{i\in{\cal I}_1}(1/\lambda_{i})+\sigma_{\epsilon}^{-2}-|{\cal I}|\beta$, $\mathcal{I}_1=\{i:i\in\mathcal{I},~P_i>0\}$, and $\beta$ be given by:
\begin{equation}
\beta=\left(2^{C_{I}}/{\displaystyle \prod_{i\in{\cal I}_1}}\lambda_{i}\right)^{\frac{1}{|{\cal I}_1|}}. \label{eq:beta}
\end{equation}
The optimal total power for this assignment is given by $P=(y-1)/\sigma_{\epsilon}^{2}$, $y\geq1$, where:
\begin{equation}
y=\left\{ \begin{array}{ll}
\sigma_{\epsilon}^{-2}/\zeta, & \text{if }\sigma_{\epsilon}^{-2}/\zeta\ge\frac{C_{EH}}{|{\cal E}|}+1\\
\frac{(C_{EH}+|{\cal E}|)/\lambda_{e}+\sigma_{\epsilon}^{-2}}{\zeta + |{\cal E}|/\lambda_{e}}, & \text{otherwise.}
\end{array}\right.
\label{eq:y}
\end{equation}
In addition, given $P$ and $y$, the optimal power allocation for a given assignment,  $\mathcal{I},~\mathcal{E}$, is: 
\begin{align}
P_{i}&=\max\{0,y(\beta-1/\lambda_{i})\}, i\in\mathcal{I}, \label{eq:Pinfo} \\
P_{i}&=
\left\{\begin{array}{ll}
\max{\left\{0,\frac{C_{EH}-(y-1)|{\cal E}|}{\lambda_{i}}\right\}}, &i=e,\\
0, &i\in \mathcal{E} \setminus \{e\}
\end{array}\right. . \label{eq:Pen}
\end{align} 
\end{theorem}
 
\begin{proof} 
The proof can be derived by following similar analysis with Section III in \cite{greencomm13}.
\end{proof} 
 
Theorem 1 consists of two parts; the first part illustrates how the optimal power, $P$, of the examined assignment can be obtained, while the second part provides its optimal power allocation. To obtain $P$ it is essential to compute parameter $\beta$ which requires finding set $\mathcal{I}_1$ containing the indices of information eigenchannels with non-zero power. These can be obtained by sequentially adding information eigenchannels to $\beta$, starting from the best one, until an eigenchannel, $k$, is found for which $\beta - 1/\lambda_k \leq 0$, $k\in\mathcal{I}$. One important observation regarding theorem 1, which emanates from remark 3, is that condition $\beta-1/\lambda_j>0$, $j=\max\{\cal{I}\}$ is necessary for a candidate assignment to be optimal. Hence, any assignments that do not satisfy this condition are immediately rejected without any further consideration. 
 
To check for infeasible solutions, we further need
to ensure that the obtained solution for $y$ yields a positive value
for $P_{e}$ and that $y\geq 1$. As $y$ is a monotonically
increasing function of $P$, by finding the best feasible value for
$y$ among all feasible assignments, $y_{opt}$, we can derive the optimal
power allocation of \eqref{eq:ProblemFA_InfP}, as outlined in algorithm \ref{alg:OptPinf}.


\begin{algorithm}[t]
	\caption{\textbf{: Optimal solution to problem \eqref{eq:ProblemFA_InfP}}} 
	\begin{algorithmic} [1] 
 	\STATE  \textbf{Init.} $y_{opt}=\infty$, $\cal{I}$$_{opt} = \varnothing$, and $e_{opt}=\varnothing$.
  \FOR{$e=1$ to $N$}
	  \FOR{$i=e+1$ to $N$}
  		\STATE Initialise fixed assignment: $\cal{I}$$=\{1,...,i-1\}\setminus\{e\}$.
  		\IF{($|\cal{I}$$| > 0$)}
  			\STATE Compute $\beta$ according to \eqref{eq:beta}.
	  		\IF{($\beta-1/\lambda_j>0$, $j=\max\{\cal{I}\}$)}
	  			\STATE Compute the value of $y$ according to \eqref{eq:y}.
		  		\IF{($(y_{opt}>y)$ AND ($y\ge1$))}
						\STATE Store the optimal solution found so far:
						\STATE $y_{opt}=y$, $\cal{I}$$_{opt}=\cal{I}$, and $e_{opt}=e$.
	  			\ENDIF
  			\ENDIF  		
  		\ENDIF
  \ENDFOR 
  \ENDFOR 
	\IF{($1\le y_{opt} < \infty$)}
		\STATE Having found optimal assignment $\cal{I}$$_{opt}$, $e_{opt}$ and $y_{opt}$ compute optimal power allocation according to \eqref{eq:Pinfo} and \eqref{eq:Pen}.
	\ELSE
		\STATE Deem problem infeasible.
	\ENDIF
\end{algorithmic}  
\label{alg:OptPinf}
\end{algorithm}

It should be emphasized that algorithm \ref{alg:OptPinf} has two very attractive characteristics: (a) it solves a nonlinear combinatorial optimization problem involving binary variables, in polynomial time, as it requires the examination of a polynomial number of fixed assignments (approximately $N^2/2$), and (b) the optimal power allocation needs to be derived only for the optimal assignment $\cal{I}$$_{opt}$, $e_{opt}$, and not for all examined fixed assignments which reduces the execution time of the algorithm. Although each assignment appears to be of computational complexity $O(N)$ due to the presence of $\sum_{i\in{\cal I}}\frac{1}{\lambda_{i}}$ and  $\prod_{i\in{\cal I}}\lambda_{i}$, we can reduce the computational complexity to $O(1)$. This can be achieved by storing the sum and product terms for fixed EH assignment, $e$, and updating their values for an increasing number of information channels $i=e+1,...,N$. Hence, the total complexity of algorithm  \ref{alg:OptPinf} is $O(N^2)$.

\section{Optimal Power Allocation with Fixed Eigenchannel Assignment}
\label{sec:powerAllocation}

In this section, we develop a waterfilling-like procedure to solve the optimal power allocation problem for a given eigenchannel assignment. Addressing this problem is essential for the development of a low complexity algorithm for problem \eqref{eq:ProblemGeneral}, discussed in Section \ref{sec:heur}; it is also used to find the optimal power allocation of the MILP approach upon derivation of the eigenchannel assignment. 

Let, $\mathcal{I}$ and $\mathcal{E}$, denote the sets of eigenchannels assigned for information and EH respectively, and $P_f$ denote a fixed value for the total power that appears in constraint \eqref{eq:ProblemGeneralb}; then, the optimal power allocation problem can be defined as:
\begin{subequations}
\label{eq:PowerAllocation}
\begin{align}
\min~ & \sum_{i\in\{\mathcal{I}\cup\mathcal{E}\}} P_i \label{eq:PowerAllocationa}\\
\text{s.t. }&\sum_{i\in\mathcal{I}} \log \left(1+\frac{P_i \lambda_i}{1+\sigma_{\epsilon}^2 P_f}  \right)\geq C_I, \label{eq:PowerAllocationb}\\
&\sum_{i\in \mathcal{E}} \big(P_i \lambda_i+\sigma_{\epsilon}^2\sum_{j\in\{\mathcal{I} \cup \mathcal{E}\}}P_j \big)\geq C_{EH}, \label{eq:PowerAllocationc} \\
&0 \leq P_i \leq P_{max},~i\in\{\mathcal{I}\cup\mathcal{E}\}. \label{eq:PowerAllocatione}
\end{align}
\end{subequations}

In formulation \eqref{eq:PowerAllocation}, the information constraint \eqref{eq:PowerAllocationb} is independent of the EH constraint as no power allocated to energy eigenchannels appears in it, while the EH constraint only depends on the total power allocated to information eigenchannels so that the two subproblems can be solved almost independently. 

\begin{algorithm}[t]
	\caption{\textbf{: Optimal power allocation for problem \eqref{eq:PowerAllocation}}} 
	\begin{algorithmic} [1] 
 	\STATE  \textbf{Init.} $P_{LB} = 0$; $P_{UB} = NP_{max}$; $\epsilon = 10^{-6}$; $f=0$.
  \WHILE{$|P_{UB} - P_{LB}|\ge\epsilon$}
  		\STATE Set $P_f = (P_{UB} + P_{LB})/2$.
			\STATE Solve the information subproblem \eqref{eq:InformationSubproblem} according to the procedure discussed in section \ref{sec:infoSub}.
			\STATE Find the optimal power allocations, $P_{i}, i\in\{\mathcal{I}\cup\mathcal{E}\}$ according to cases 1-4 discussed in section \ref{sec:enerSub}.
  		\IF{(problem feasible)AND($\sum_{i\in\{\mathcal{I}\cup\mathcal{E}\}}P_i \leq P_f$)}
				\STATE Set $P_{UB} = P_f$; $P^{*}_{i}=P_{i}$; $f=1$.
			\ELSE  		
				\STATE Set $P_{LB} = P_f$.
  		\ENDIF
	\ENDWHILE
	\IF{$f=0$}
  		\STATE Problem \eqref{eq:PowerAllocation} is infeasible.
  \ENDIF
\end{algorithmic}  
\label{alg:PAFA}
\end{algorithm}

Note that by fixing the total power in the information constraint of the power allocation problem, a bisection procedure has to be followed to obtain the optimal total power. This procedure is summarized in algorithm \ref{alg:PAFA}, where the solution of the information and energy subproblems is discussed in sections \ref{sec:infoSub} and \ref{sec:enerSub}, respectively. The computational complexity of algorithm \ref{alg:PAFA}, is equal to $O(\log(1/\epsilon)N\log N)$, where $\log(1/\epsilon)$ is the estimated number of iterations of the bisection procedure with a stopping tolerance $\epsilon$, while $O(N\log N)$ is the computational complexity for solving the information subproblem.

\subsection{Information subproblem solution}
\label{sec:infoSub}

For given assignment and fixed total power in the information constraint, the information subproblem is expressed as:
\begin{subequations}
\label{eq:InformationSubproblem}
\begin{align}
\min &~ \sum_{i\in\mathcal{I}} P_i \label{eq:InformationSubproblema}\\
\text{s.t.}\;\;&\sum_{i\in\mathcal{I}} \log \left(1+\frac{P_i \lambda_i}{1+\sigma_{\epsilon}^2 P_f}  \right)\geq C_I, \label{eq:InformationSubproblemb}\\
&0 \leq P_i \leq P_{max},~i\in\mathcal{I}. \label{eq:InformationSubproblemd}
\end{align}
\end{subequations}
The solution of problem \eqref{eq:InformationSubproblem} is summarized in theorem 2. 
\begin{theorem}
Let $\theta_i =(1+\sigma_{\epsilon}^{2}P_{f})/\lambda_{i}$. The solution of the information subproblem \eqref{eq:InformationSubproblem}, if the problem is feasible ($\text{i.e.,} \sum_{i\in\mathcal{I}} \log (1+P_{max}/\theta_i)\geq C_I$),  is given by:
\begin{equation}
\label{eq:PinfoOptimal}
P_{i}^{*}=\left\{ \begin{array}{ll}
0, & \nu<\theta_i \\
\nu-\theta_i, & \theta_i\leq \nu \leq \theta_i+P_{max}\\
P_{max}, & \text{otherwise},
\end{array}\right.
\end{equation}
where the Lagrange multiplier $\nu$ is derived from: 
\begin{equation}
\label{eq:PinfoNuOptimal}
\nu = \rho^{(1/|\mathcal{I}_1|)}, \text{ and }\rho = \frac{2^{C_I}\prod_{i\in\mathcal{I}_1}(\theta_i)}{\prod_{i\in\mathcal{I}_2}(1+P_{max}/\theta_i)}.
\end{equation}
Sets $\mathcal{I}_1$ and $\mathcal{I}_2$ are defined as $\mathcal{I}_1=\{i:0<P_{i}^{*}<P_{max}\}$ and  $\mathcal{I}_2=\{i:P_{i}^{*}=P_{max}\}$. 
\end{theorem}
\begin{proof}
The proof is given in Appendix A.
\end{proof}

Theorem 2 indicates that the computation of the optimal power allocation can be easily obtained from \eqref{eq:PinfoOptimal}, provided that the value of $\nu$ satisfying  \eqref{eq:PinfoNuOptimal} is found.  This can be achieved if the eigenchannels in sets $\mathcal{I}_1$ and $\mathcal{I}_2$ are found. Towards this direction, first we need to sort the $2|\mathcal{I}|$ values $\theta_i$ and $\theta_i + P_{max}$ in ascending order; let us assume that this order is $q_1 \leq q_2 \leq$...$\leq q_{2|\mathcal{I}|}$. If $\nu\in[q_j,q_{j+1}]$, sets $\mathcal{I}_1$ and $\mathcal{I}_2$ can be constructed and a candidate value, $\nu_c$, can be computed from \eqref{eq:PinfoNuOptimal}; in case $\nu_c\in[q_j,q_{j+1}]$ then it is optimal, otherwise the subsequent range of values needs to be examined, i.e., $[q_{j+1},q_{j+2}]$. For example, if $|\mathcal{I}|=2$, $\theta_1=1$, $\theta_2=2$, $P_{max}=0.5$, and $\nu=2.3$, with $q_i=\{1,1.5,2,2.5\}$, then, according to \eqref{eq:PinfoNuOptimal}, it will be true that $\mathcal{I}_1=\{2\}$ and  $\mathcal{I}_2=\{1\}$, as $\nu - \theta_1 > P_{max}$ and $0 < \nu - \theta_2 < P_{max}$. 

A computationally efficient method to obtain the optimal value of $\nu$ is to sequentially keep updating $\rho$ and $r=|\mathcal{I}_1|$ for cheap computation of $\nu_c$, while examining new regions. When examining region $[q_j,q_{j+1}]$, two cases can occur: (a) addition of some eigenchannel, say $i$, into set $\mathcal{I}_1$ (when $q_j = \theta_i$), and (b) addition of eigenchannel $i$ into set $\mathcal{I}_2$ (when $q_j = \theta_i+P_{max}$). The first case,  requires addition of element $i$ into set $\mathcal{I}_1$ so that $r$ has to be increased by one, and $\rho$ has to be multiplied by $q_j=\theta_i$. The second case,  requires deletion of eigenchannel $i$ from set $\mathcal{I}_1$ and addition into  $\mathcal{I}_2$ so that $r$ has to be decreased by one; updating $\rho$ requires division by $\theta_i$, due to the deletion of eigenchannel $i$ from set $\mathcal{I}_1$, and division by  $(1+P_{max}/\theta_i)$, due to addition of element $i$ into $\mathcal{I}_2$, yielding division by $\theta_i(1+P_{max}/\theta_i)=\theta_i+P_{max}=q_j$. To summarize, starting from $\rho=2^{C_I}$ and $r=0$ the two cases require setting $r=r+1$, $\rho = \rho q_j$, and $r=r-1$, $\rho = \rho/ q_j$, respectively; once $r$ and $\rho$ are updated, we compute $\nu_c=\rho^{1/r}$ and check to see whether $\nu_c\in[q_j,q_{j+1}]$ in which case we stop, otherwise we examine the subsequent range of values. The complexity of solving the information subproblem is $O(N\log(N))$, due to the sorting of the $q_i$ values.

\subsection{EH subproblem solution}
\label{sec:enerSub}

Once the information subproblem is solved, the optimal solution of \eqref{eq:PowerAllocation} is obtained by solving the \emph{EH subproblem}
\begin{subequations}
\label{eq:EHsubProblem}
\begin{align}
\min~ & \sum_{i\in\mathcal{E}} P_i \label{eq:EHsubProblema}\\
&\sum_{i\in \mathcal{E}} \big(P_i \lambda_i+\sigma_{\epsilon}^2\big(\sum_{j\in\mathcal{E}}P_j + P_I\big)\big)\geq C_{EH}, \label{eq:EHsubProblemb} \\
&0 \leq P_i \leq P_{max},~i\in\mathcal{E}, \label{eq:EHsubProblemc}
\end{align}
\end{subequations}
where $P_I=\sum_{i\in\mathcal{I}}P^*_{i}$. The solution of \eqref{eq:EHsubProblem} falls within the following four cases:

\textbf{Case 1:} The minimum power allocated for the information eigenchannels, $P_I$ suffices to satisfy the EH constraint, i.e., $|\mathcal{E}|P_I \sigma_{\epsilon}^{2} \geq C_{EH}$. Hence, no power needs to be allocated to EH eigenchannels.

\textbf{Case 2:} The minimum power allocated for the information eigenchannels and the total power allocated for the EH eigenchannels suffices to satisfy the EH constraint, i.e., $P_{max}\sum_{i\in\mathcal{E}}\lambda_i + |\mathcal{E}|(P_I + |\mathcal{E}|P_{max})\sigma_{\epsilon}^{2} \geq C_{EH}$. Because the EH constraint \eqref{eq:PowerAllocationc} is linear, power allocation is performed by sequentially filling EH eigenchannels, starting from the one with the largest eigenvalue to the one with the smallest eigenvalue, until the EH constraint is satisfied.  

\textbf{Case 3:} The total power allocated to the information and EH eigenchannels suffices to satisfy the constraints, i.e., 	$|\mathcal{E}|N P_{max}\sigma_{\epsilon}^{2} + P_{max}\sum_{i\in\mathcal{E}}\lambda_i \geq C_{EH}$. In this case all EH eigenchannels are allocated maximum power, i.e., $P_{i_2}=P_{max}$, $i_2\in\mathcal{E}$, and extra power is arbitrarily added to information channels until the EH constraint is satisfied.

\textbf{Case 4:} The total power allocated to the information and EH eigenchannels is not sufficient to satisfy either the information or the EH constraint, hence the problem is infeasible. 

Note that the above four cases are examined sequentially until one of them is satisfied, yielding the optimal solution to \eqref{eq:PowerAllocation}. The complexity of solving the EH subproblem is $O(N)$.

\section{Low-complexity heuristic}
\label{sec:heur}

In general, the solution methods proposed for problem \eqref{eq:ProblemGeneral} in Sections \ref{sec:MISOCP} and \ref{sec:MILP} have exponential complexity due to the nonlinear and combinatorial nature of the problem, and hence not suitable for real-time execution. In this section we propose a polynomial complexity heuristic that provides suboptimal results, but is suitable for real-time execution. 

The heuristic is based on the observation that one constraint usually dominates over the other, acquiring the ``best'' eigenchannels, i.e., those with the largest eigenvalues. If the information constraint is dominant, the first $k$ consecutive eigenchannels will be assigned for information decoding and the rest for EH, where $k$ is a value to be determined. Nonetheless, if the EH constraint is dominant, then apart from the ``best'' eigenchannels, the ``worse'' eigenchannels could be given for EH, with power only allocated in the best ones. The reason is that empty eigenchannels ($P_i=0$) are only useful for EH. Note that in both cases, information eigenchannels remain grouped, and for this reason we call the heuristic group eigenchannel assignment (GEA).

\begin{algorithm}[t]
	\caption{\textbf{: Group Eigenchannel Assignment} }
	\begin{algorithmic} [1] 
 	\STATE \textbf{Init.} Set $P=\infty$. 
 	\STATE Compute $P^{inf}_i$, $i\in\mathcal{N}$, $\{\mathcal{I}^{inf},~\mathcal{E}^{inf}\}$,  using algorithm \ref{alg:OptPinf}. 
	\IF{\text{all}($P^{inf}_{i} \leq P_{max}$)}   	
		\STATE $P_i^* = P_i^{inf}, i\in\mathcal{N}$; $\mathcal{I} = \mathcal{I}^{inf}$; $\mathcal{E} = \mathcal{E}^{inf}$.
	\ELSE
  	\FOR{$i\in\mathcal{N}$}   
				  \STATE Compute the minimum number of information eigenchannels $N_{i}^{min}$, with $\mathcal{I} = \{i,...,N\}$.
			  	\FOR{$j = i+N_{i}^{min} - 1,...,N$}
			 			\STATE Set $\mathcal{I}_c = \{i,...,j\}$.
			 			\STATE Set $\mathcal{E}_c = \{1,...,i-1,j+1,...,N\}$.
						\STATE Obtain a lower bound to the total power, $P^{inf}_{LB}=(y-1)/\sigma_{\epsilon}^{2}$, where $y$ is computed from \eqref{eq:y} according to theorem 1. 
						\IF{($P^{inf}_{LB}<P$)}   
							\STATE  Solve problem \eqref{eq:PowerAllocation} to obtain $P_i^{OPA}, i\in \mathcal{N}$. 
							\IF{($\sum_i P_i^{OPA}<P$)}
								\STATE Store current solution: $P = \sum_i P_i^{OPA}$;\\
								$P_i^* = P_i^{OPA}, i\in\mathcal{N}$; $\mathcal{I} = \mathcal{I}_c$; $\mathcal{E} = \mathcal{E}_c$.
							\ENDIF
						\ENDIF
					\ENDFOR
		\ENDFOR
		\ENDIF
\end{algorithmic}  
\label{alg:GEA}
\end{algorithm}

GEA heuristic is outlined in algorithm \ref{alg:GEA}. First, algorithm \ref{alg:OptPinf} is examined as a way to optimally solve cases of low transmitted power, i.e. with $P_i^*\leq P_{max}$, $i\in\mathcal{N}$. If that is not the case, GEA implements the main idea of assigning information eigenchannels in a group, as well as two strategies that significantly reduce the total number of power allocation problems that need to be solved. The first strategy estimates the minimum number of eigenchannels required to satisfy the information constraint, $N_{i}^{min}$, given that $i$ is the largest information eigenchannel.  $N_{i}^{min}$ is the first number that satisfies the information constraint with $\sigma_{\epsilon}^2=0$, i.e.,  $\sum_{m=i}^{N_{i}^{min}} \log(1+P_{max}\lambda_m)>C_I$. The second strategy solves the fixed assignment problem for $P_{max}=\infty$ to obtain a lower bound to the total power, $P^{inf}_{LB}$, according to theorem 1, so that the standard fixed assignment problem is solved only when the currently best solution $P$ is less than $P^{inf}_{LB}$. Note that if any subproblem is infeasible the returned total power is equal to infinity. 

GEA involves $N(N-1)/2$ iterations to account all considered eigenchannel combinations, while the maximum complexity for each assignment is $O(N\log(N)\log(1/\epsilon))$ which is the complexity to solve \eqref{eq:PowerAllocation}. Note that neither of the introduced strategies increase the complexity of the algorithm as the complexity of solving \eqref{eq:PowerAllocation} is larger. Hence, the computational complexity of the algorithm is $O(N^3\log(N)\log(1/\epsilon))$.


\section{Numerical results}
\label{sec:nume}

\begin{figure*}[t]
\centering
\subfigure[]{
  \includegraphics[width=0.47\linewidth]{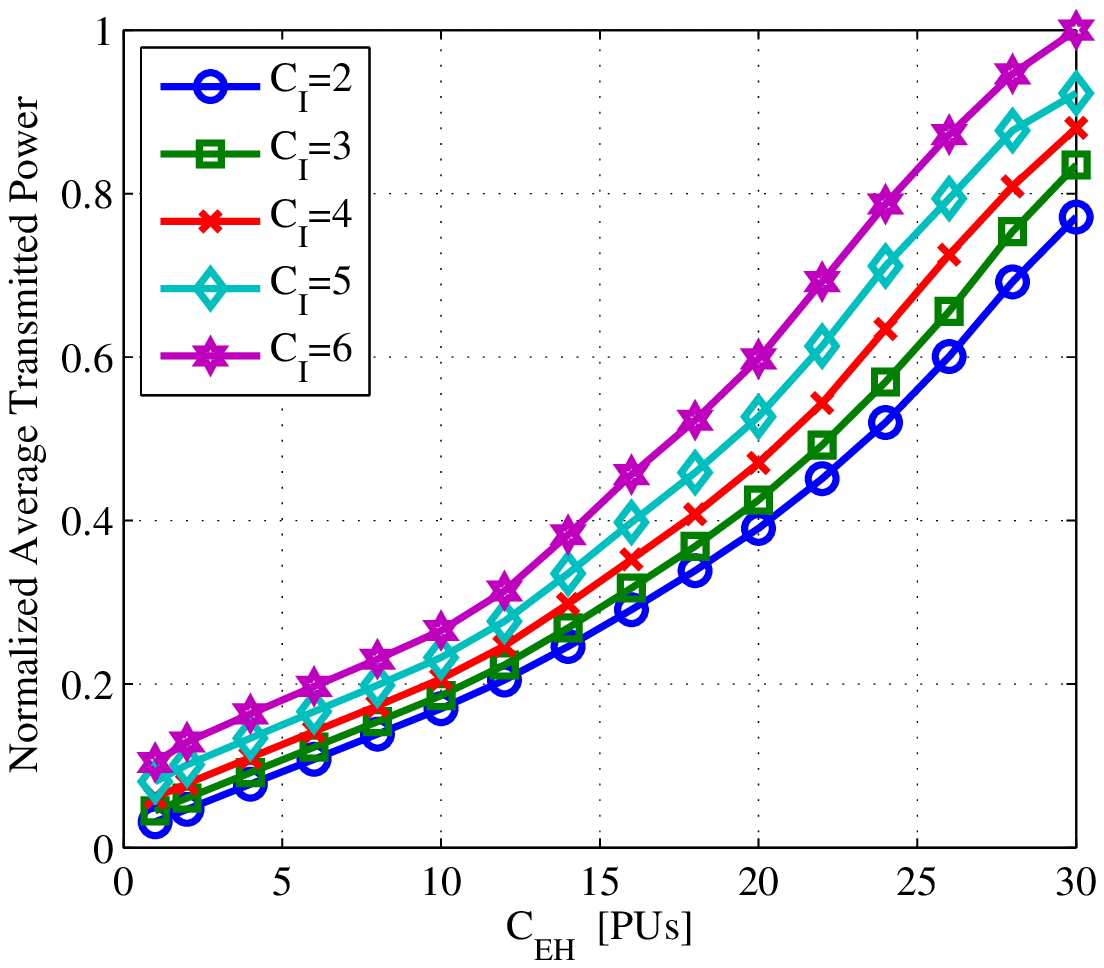}
	\label{fig:fig1}
 }
\subfigure[]{
  \includegraphics[width=0.47\linewidth]{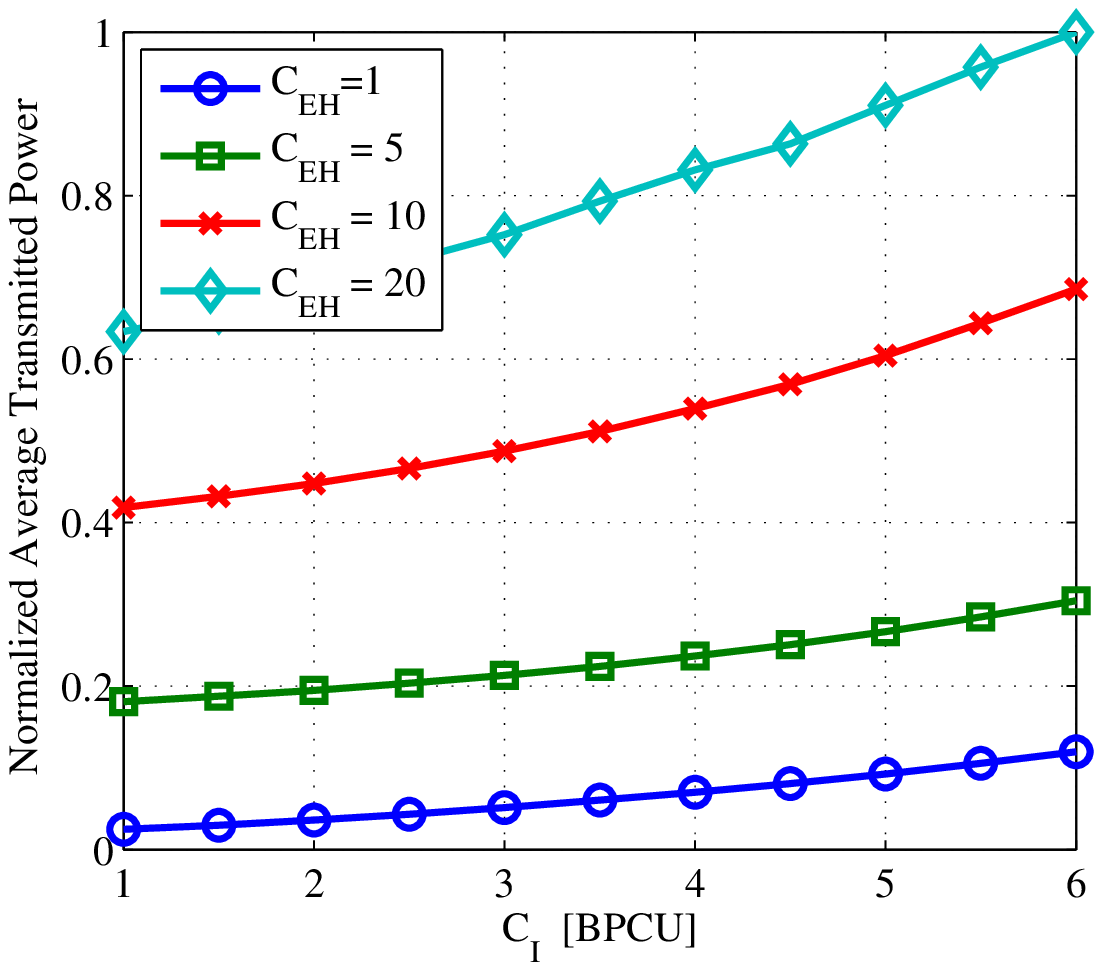}
	\label{fig:fig2}
 }
\subfigure[]{
  \includegraphics[width=0.47\linewidth]{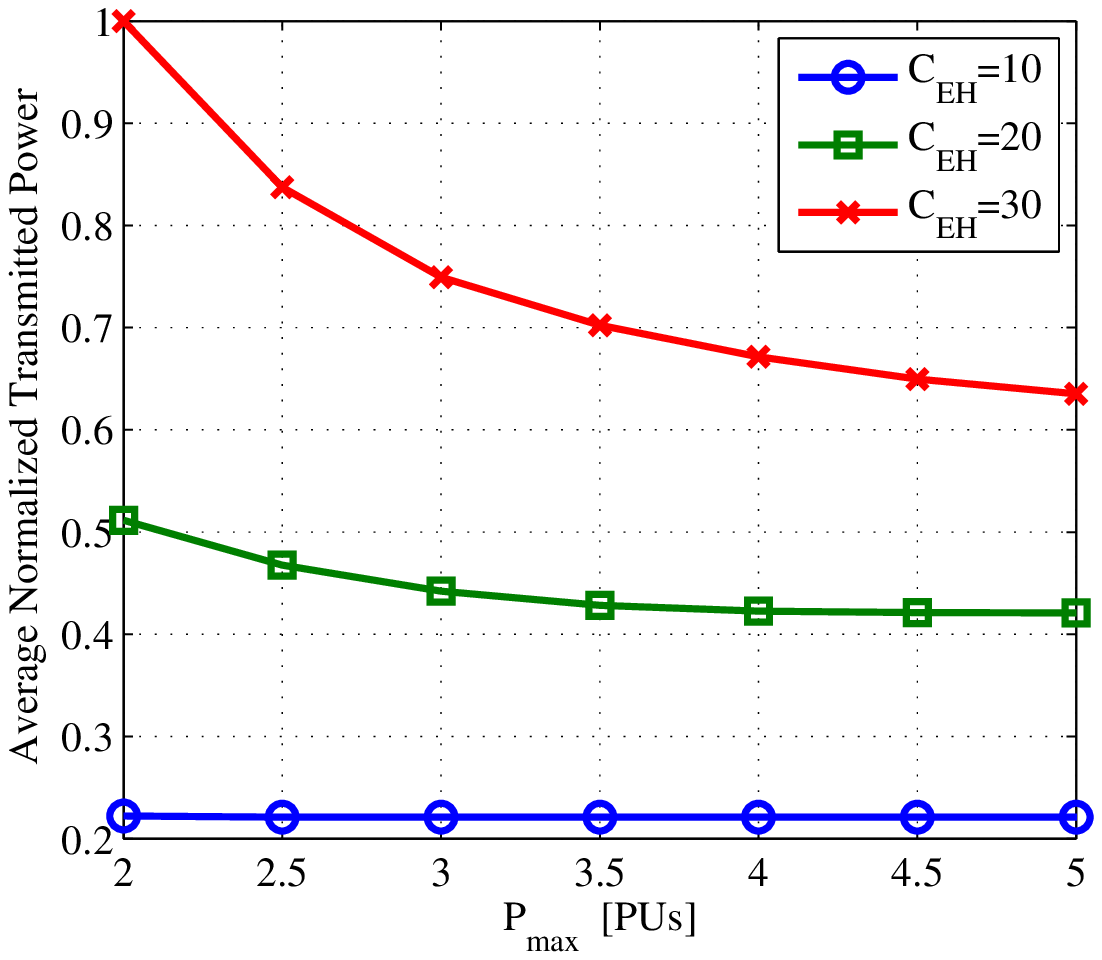}
	\label{fig:fig3}
 }
\subfigure[]{
  \includegraphics[width=0.47\linewidth]{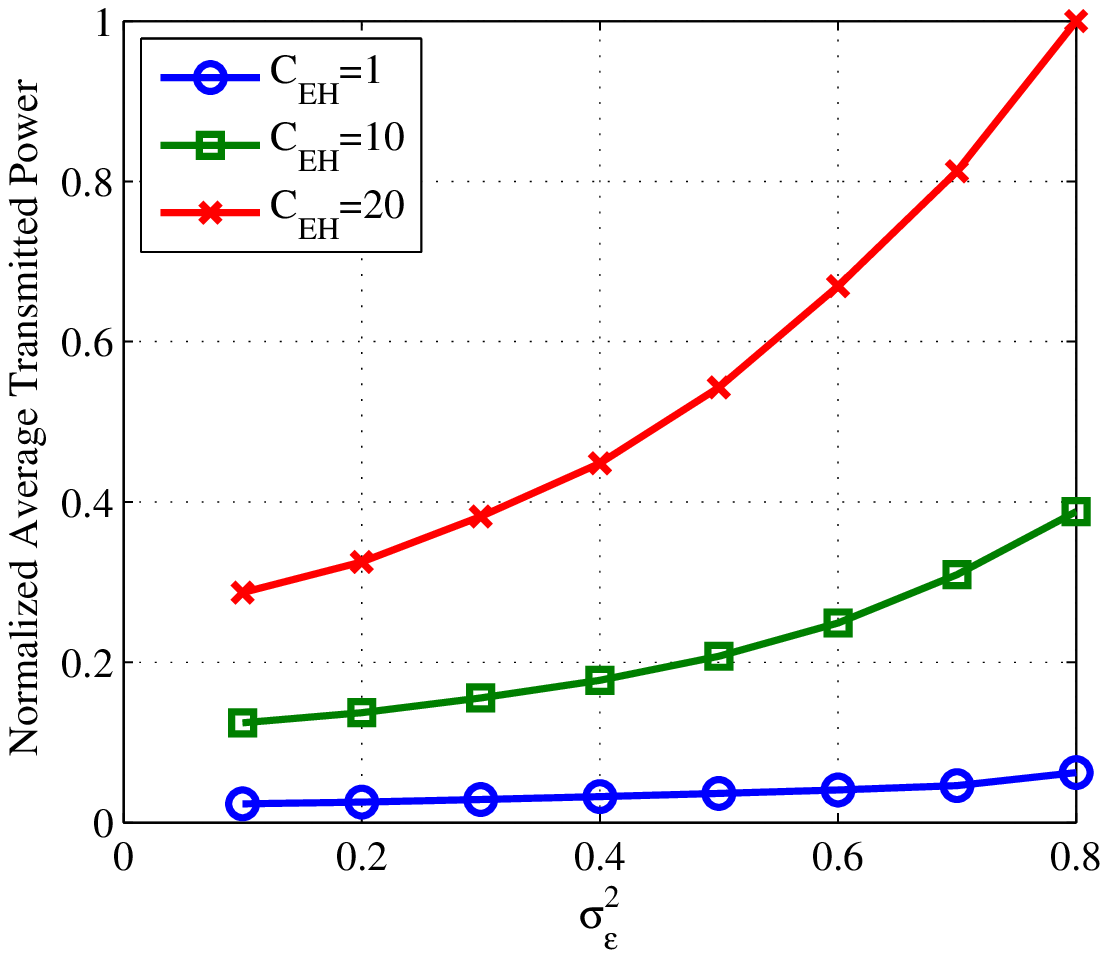}
	\label{fig:fig4}
 }
\vspace{-0.3cm}
\caption{Effect of different parameters on the normalized average transmitted power: (a) $C_{EH}$, (b) $C_{I}$, (c) $P_{max}$, and (d) $\sigma^2_{\epsilon}$.} 
\end{figure*}

Computer simulations are carried-out in order to evaluate the performance of the developed algorithms. Since the scope of the paper is to highlight the theoretical idea of SWIPT in point-to-point MIMO setup, only small scale fading is considered for the wireless medium. Specifically, for each different experimental configuration, the results are averaged over several randomly generated instances of $N_R\times N_T$ channel matrix in which the entries of the channel are independent and identically distributed ZMCSCG random variables with unit variance. Thus, the information rate is measured in Bits Per Channel Use (BPCU) and power is measured in Power Units (PUs). Mathematical modeling of MISOCP and MILP formulations was done using the Gurobi optimization solver \cite{Gurobi}.

\subsection{The effect of harvesting and information constraints}

To evaluate the effect of harvesting and information constraints, the optimal MISOCP solution in \eqref{eq:SOCP2} is invoked to solve several instances of the original problem under different configuration setups. In all instances, an $8 \times 8$ MIMO channel is considered (thus $N=8$) and the energy conversion efficiency factor is set to $\delta=0.3$. In each figure, the performance is shown in a relative manner by using the measure of the Normalized Average Transmission Power, which for each setup is defined as the ratio between the average of the objective value (of this specific setup) to the maximum average value of the objective over all the setups which are illustrated in the figure.

In Fig. \ref{fig:fig1} and Fig. \ref{fig:fig2}, we plot the average normalized transmitted power of MISOCP versus the RF EH constraint ($C_{EH}$) and information rate constraint ($C_I$), respectively, when $\sigma^2_{\epsilon}=0.1$ and $P_{max}=2$ PUs. Each point in the figures is the normalized average value of the total power consumption calculated over 1000 problem instances. As expected, the required transmitted power increases as $C_{EH}$ and/or $C_{I}$ take higher values because the constraints are harder to satisfy. One interesting observation that was extracted by analyzing MISOCP's eigenchannel assignment output, is that a significant portion of energy constraint $C_{EH}$ is covered from information eigenchannels by exploiting the interference term in \eqref{rate}.

In Fig. \ref{fig:fig3} we plot the impact of the peak power constraint $P_{max}$ on the performance of MISOCP for different RF EH thresholds and $C_I=10$ BPCU  and $\sigma^2_{\epsilon}=0.1$. As the bound $P_{max}$ increases, the inherent spatial diversity of the system is exploited in a more efficient way by the optimal power allocation policy. Thus, the performance of the system is improved, despite the fact that potentially more power is available to the system. For any given $C_{EH}$ and $C_I$ value, there is a point where any further increase in $P_{max}$ does not affect the performance of the system. This is because there is no eigenchannel in which the upper power bound constraint is active, \emph{e.g} for all eigenchannels the power allocated amount is strictly less than $P_{max}$. Clearly, this phenomenon emerges for larger values of $P_{max}$ as the constraints become harder.

In Fig. \ref{fig:fig4} we show the impact of channel estimation error $\sigma^2_\epsilon$ on the transmitted power for different RF EH thresholds $C_{EH}=\{2,10,20\}$. As can be seen the value of $\sigma^2_\epsilon$ significantly affects the total power consumption. For a fixed value of $C_{EH}$, a higher estimation error means that more power is required to guarantee the information rate constraint. The increase becomes more intense for higher values of the energy constraint.

\subsection{Comparison between MISOCP and MILP}

As it was concluded by extensive simulation results, the performance of MILP is perfectly matched to MISOCP in terms of average transmission power. This is an interesting result since the two approaches aim to solve the original problem from two completely different viewpoints. However, one important issue that should be taken into consideration is the execution time of each solution algorithm. In Fig. \ref{fig:figTime}, the execution time of MILP and MISOCP with respect to the size of the problem is shown, \emph{e.g} the number of eigenchannels, for two different setup configurations. The bound for the approximation error in MILP was set equal to $e_{max}=10^{-3}$. As can be seen, the execution of MILP is more robust compared to the execution time of MISOCP which becomes non-practical as $N$ increases. Having algorithms that provide robust execution times for large $N$, such as the MILP algorithm, is particularly important in obtaining the optimal solution in massive MIMO configurations where hundreds of antenna elements are present.

\begin{figure}[t]
\centering
\includegraphics[width=\linewidth]{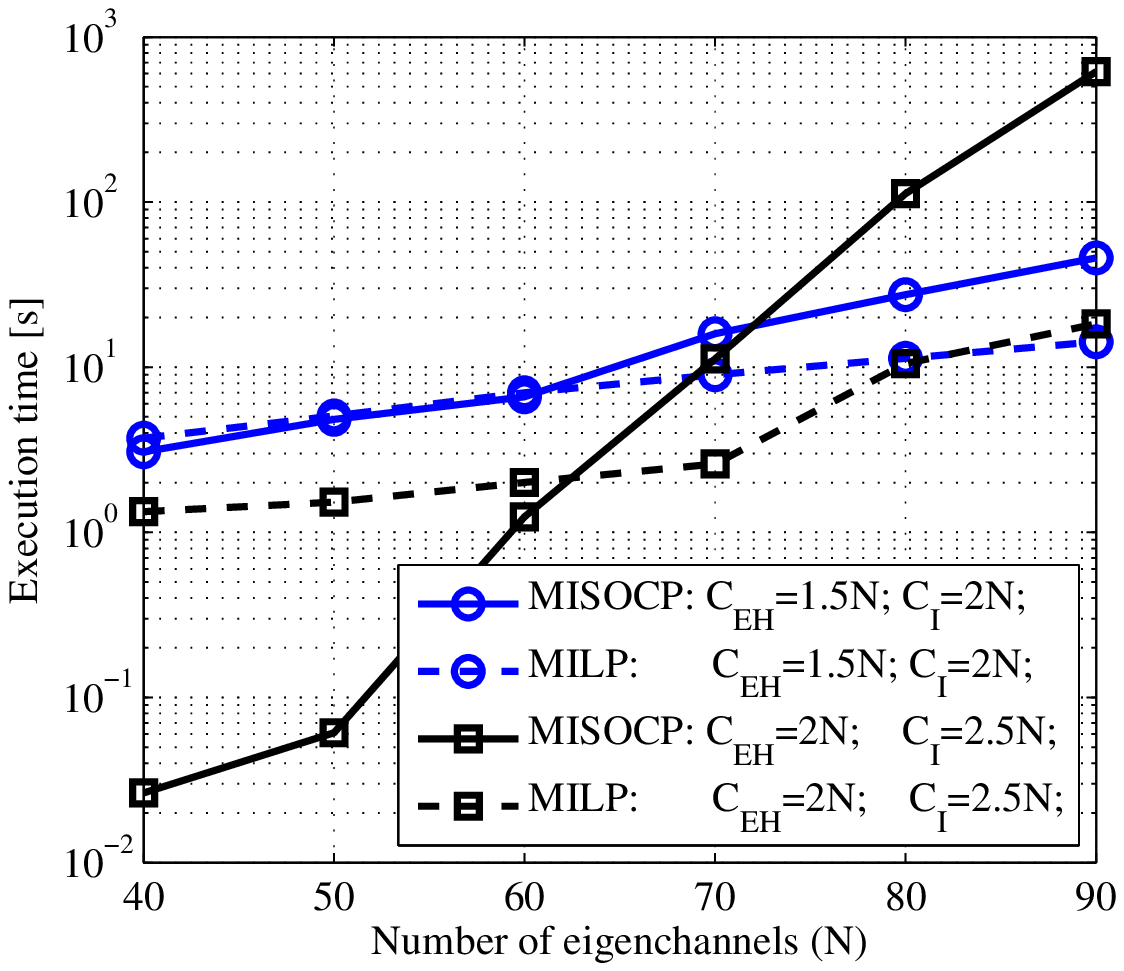}
\vspace{-0.5cm}
\caption{Execution time versus number of eigenchannels for different RF EH and data constraints.}
\label{fig:figTime}
\end{figure}

\subsection{Comparison between MISOCP and GEA}

In Fig. \ref{fig:fig5}, the performance of the Group Eigenchannel Assignment (GEA) algorithm and MISOCP is compared by averaging 10000 problem instances. An $8 \times 8$ MIMO channel is considered with $P_{max}=2$ PUs, $\sigma^2_\epsilon=0.1$ and several different values $C_{EH}$ and $C_I$. As can be seen, GEA has excellent performance, achieving results within 5\% from the optimal MISOCP solutions for all the examined configuration setups. This is a quite significant result if we take into consideration that GEA is a low-complexity algorithm, while MISOCP has exponential complexity. The small gap in performance is because GEA cannot output an eigenchannel assignment where information eigenchannels are not sequential, which appear to be the optimal assignment in certain cases. In Fig. \ref{fig:fig6}, the feasibility of the two algorithms is compared against the $C_{EH}$ parameter for three different values of $C_{I}$. Again, a small gap exists in GEA's feasibility which becomes more notable when the RF EH and/or information constraints become more stringent. This observation is inline with the results in Fig. \ref{fig:fig5} where the performance of GEA deteriorates as $C_{EH}$ and $C_{I}$ take higher values.

\begin{figure*}[t]
\centering
\subfigure[]{
  \includegraphics[width=0.43\linewidth]{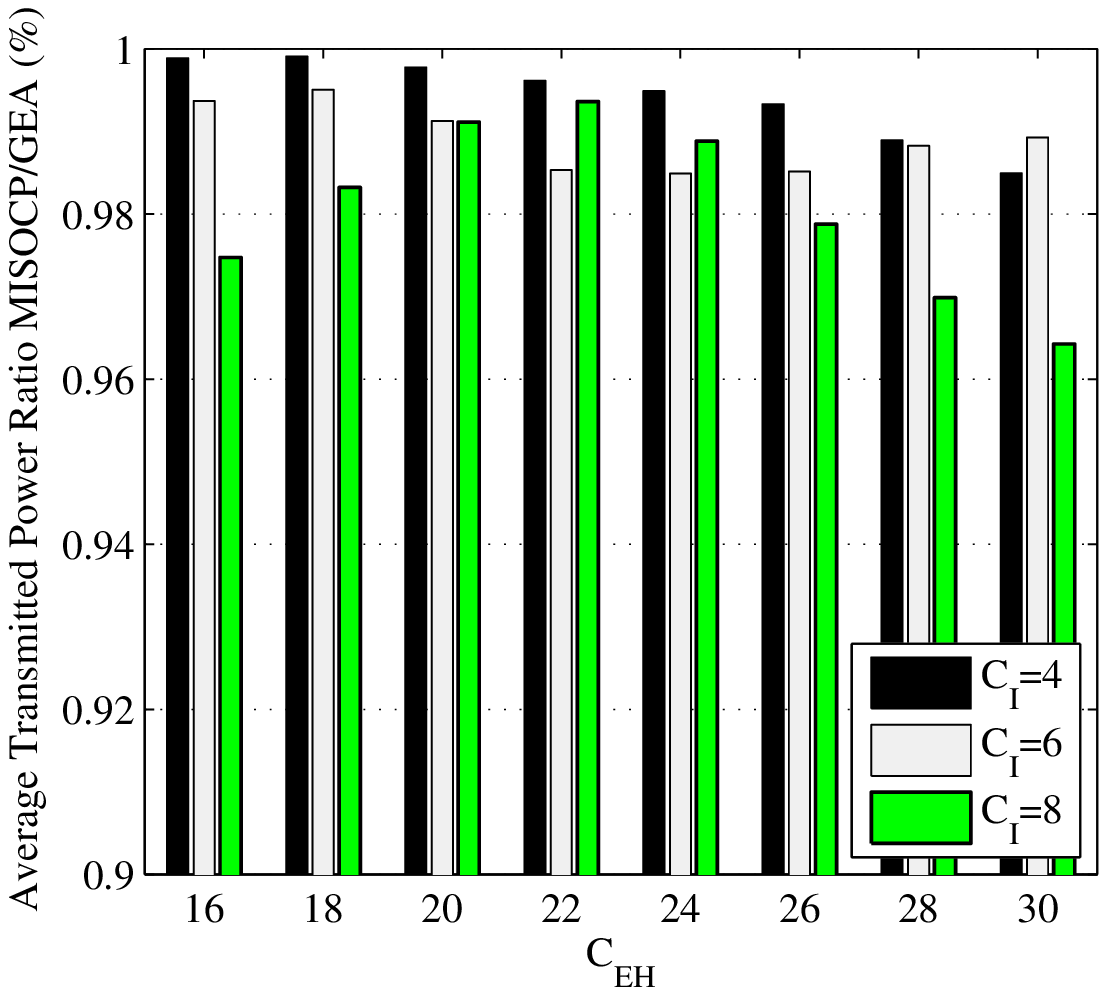}
	\label{fig:fig5}
 }
\subfigure[]{
  \includegraphics[width=0.43\linewidth]{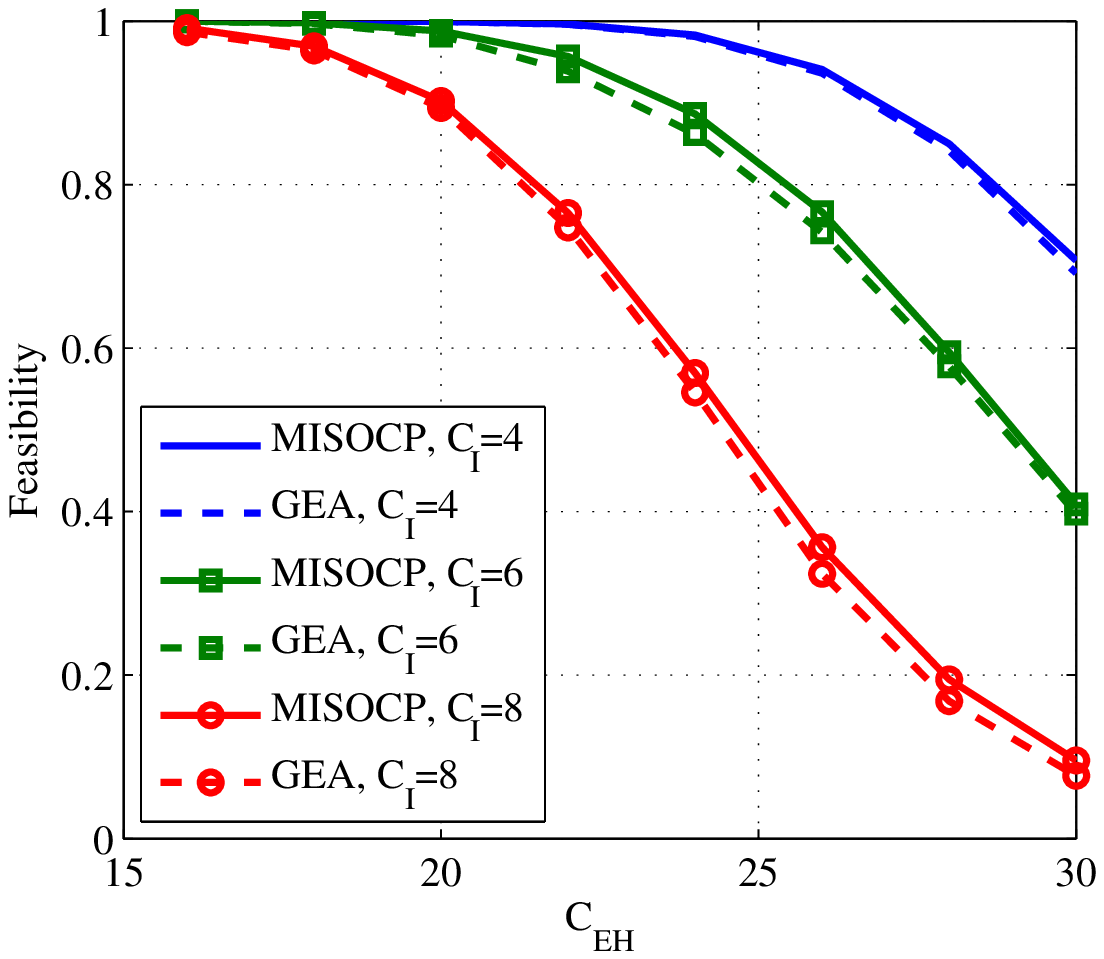}
	\label{fig:fig6}
 }
\vspace{-0.3cm}
\caption{Relative performance between MISOCP and GEA in terms of: (a) optimality, and (b) feasibility.} 
\end{figure*}

\section{Conclusions}

In this paper, we have theoretically investigated SWIPT in the spatial domain for a MIMO channel with RF EH capabilities. By using SVD decomposition for the wireless channel under an imperfect channel knowledge, the proposed technique uses the eigenchannels to convey either information or energy, with the goal of minimizing the overall transmitted power subject to information and energy  constraints. Although the examined problem is non-linear and combinatorial, MISOCP and MILP formulations have been developed that provide optimal and very close to optimal solutions, respectively.  Using Lagrange theory, a waterfilling-like procedure for the optimal power allocation is derived when the eigenchannel assignment is known. Finally, a polynomial complexity algorithm is presented  which produces a near-optimal solution for a wide range of parameter configurations. In this paper, we have laid the theoretical foundation of the spatial domain SWIPT which can be a promising technology especially for millimeter wave communication systems. In future work,  practical implementations of such systems need to be investigated.  

\section*{Appendix A: Proof of Theorem 2}

In order to prove the theorem, we consider the Karush-Kuhn-Tucker (KKT)  optimality conditions, as \eqref{eq:InformationSubproblem} is a convex optimization problem \cite{boyd}. Let $\mu_i^L$ and $\mu_i^U$ denote the Lagrange multipliers corresponding to the constraints $P_i\geq 0$ and $P_i\leq P_{max}$, $i\in\mathcal{I}$, respectively and $\nu$ the Lagrange multiplier corresponding to the information constraint \eqref{eq:InformationSubproblemb}. Then, the KKT conditions can be written as: 
\begin{subequations}
\label{eq:KKT}
\begin{align}
&\sum_{i\in\mathcal{I}} \log \left(1+P_i/\theta_i\right)\geq C_I, \label{eq:KKTa}\\
&0 \leq P_i \leq P_{max},~i\in\mathcal{I},  \label{eq:KKTb}\\
&\nu \left(\sum_{i\in\mathcal{I}} \log \left(1+P_i/\theta_i\right) - C_I\right)  = 0, \label{eq:KKTc}\\
&\mu_i^L P_i = 0, i\in\mathcal{I}, \label{eq:KKTd}\\
&\mu_i^U (P_i-P_{max}) = 0, i\in\mathcal{I}, \label{eq:KKTe}\\
&\nu\geq0,~\mu_i^L\ge0,~\mu_i^U\geq 0,~i\in\mathcal{I}, \label{eq:KKTf}\\
&1 - \nu/(\theta_i+P_i) - \mu_i^L + \mu_i^U = 0,~i\in\mathcal{I}, \label{eq:KKTg}
\end{align}
\end{subequations}
where \eqref{eq:KKTa}-\eqref{eq:KKTb} are the primal feasibility constraints, \eqref{eq:KKTc}-\eqref{eq:KKTe} are the complementary slackness conditions, \eqref{eq:KKTf} are the dual feasibility constraints, while \eqref{eq:KKTg} is obtained by setting the derivative of the Lagrange function with respect to $P_i$ to zero. Based on the KKT conditions \eqref{eq:KKT} a number of restrictions hold: 

\noindent \textbf{Restriction 1:} Because $C_I>0$, the minimization of $\sum_{i\in \mathcal{I}}P_i$ implies that constraint \eqref{eq:KKTa} is always satisfied with equality at the optimal solution; otherwise, some $P_i$ could be further reduced to achieve equality for the information constraint and also reduce the objective function value. Hence, \eqref{eq:KKTc} is always satisfied  and  $\nu\ge 0$. 

\noindent \textbf{Restriction 2:} Eqs. \eqref{eq:KKTd} and \eqref{eq:KKTe} imply that when $P_i=0$ it is true that $\mu_i^U = 0$, when $0<P_i<P_{max}$ it is true that $\mu_i^L = 0$ and $\mu_i^U = 0$, and when $P_i=P_{max}$ it is true that $\mu_i^U = 0$.

\noindent \textbf{Restriction 3:} According to the values of $\mu_i^L$ and $\mu_i^U$ we can identify three cases for the value of $\nu$ obtained from \eqref{eq:KKTg}:
\begin{enumerate}
	\item [3.1] Case $\mu_i^U=0$ and $\mu_i^U\geq 0$ ($P_i=0$): 	$\mu_i^L= 1-\nu/\theta_i \geq 0 \Rightarrow \nu \le \theta_i$.
	\item [3.2] Case $\mu_i^U=\mu_i^U=0$ ($0<P_i<P_{max}$): $\nu =\theta_i + P_i>0$.
	\item [3.3] Case $\mu_i^L=0$ and $\mu_i^U\ge0$ ($P_i=P_{max}$): $\nu \geq \theta_i+P_{max}>0$.
\end{enumerate}

Restriction 3, guarantees  that \eqref{eq:PinfoOptimal} is satisfied for all information eigenchannels and that $\nu>0$ if $P_i>0$, for some $i\in\mathcal{I}$, which is always true since $C_I>0$.

Notice that the contribution of eigenchannels with $P_i=0$ or $P_i=P_{max}$ is not depended on $\nu$ (it is equal to 0 and $\log(1+P_{max}/\theta_i)$, respectively), while for $0<P_i<P_{max}$,  the contribution is equal to $\log(\nu/\theta_i)$ (obtained from substitution of $P_i = \nu -\theta_i$ into the associated log term of the information constraint). Substituting all contributing channels into \eqref{eq:KKTa} yields:
\begin{equation*}
\sum_{i\in\mathcal{I}_1}\log(\nu/\theta_i) + 
\sum_{i\in\mathcal{I}_2}\log(1+ P_{max}/\theta_i) = C_I
\end{equation*}
where $\mathcal{I}_1$ and $\mathcal{I}_2$ are defined in Theorem 2. Rewriting the logarithmic terms in product form yields:
\begin{equation*}
\log\left\{\frac{\nu^{|\mathcal{I}_1|}}{\prod_{i\in\mathcal{I}_1}\theta_i}\prod_{i\in\mathcal{I}_2}(1+ P_{max}/\theta_i)\right\} = C_I
\end{equation*}
which after simple algebra yields \eqref{eq:PinfoNuOptimal}. This completes the proof.

\begin{IEEEbiography}[{\includegraphics[width=1in,height=1.25in,clip,keepaspectratio]{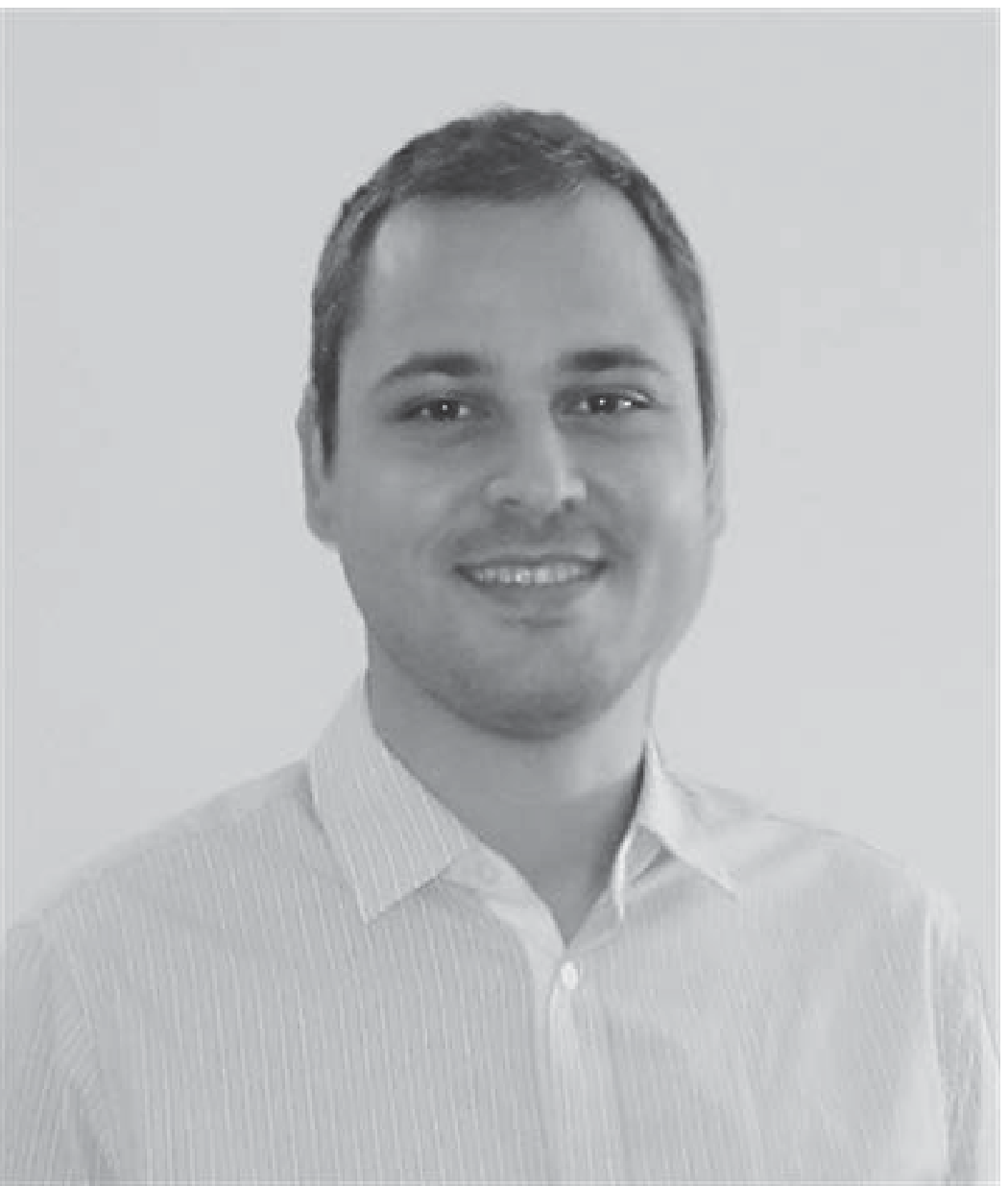}}]%
{Stelios Timotheou} (S'04-M'10) received a B.Sc. from the Electrical and Computer Engineering (ECE) School of the National Technical University of Athens, and an M.Sc. and Ph.D. from the Electrical and Electronic Engineering Department of Imperial College London. He is currently a Research Associate at the KIOS Research Center for Intelligent Systems and Networks of the University of Cyprus (UCY). In previous appointments, he was a Visiting Lecturer at the ECE Department of UCY, a Research Associate at the Computer Laboratory of the University of Cambridge and a Visiting Scholar at the Intelligent Transportation Systems Center \& Testbed, University of Toronto. His research focuses on the modelling and system-wide solution of problems in complex and uncertain environments that require real-time and close to optimal decisions by developing optimisation, machine learning and computational intelligence techniques. Application areas of his work include communication systems, intelligent transportation systems, disaster management and neural networks. He is a member of the IEEE and the ACM.
\end{IEEEbiography}

\begin{IEEEbiography}[{\includegraphics[width=1in,height=1.25in,clip,keepaspectratio]{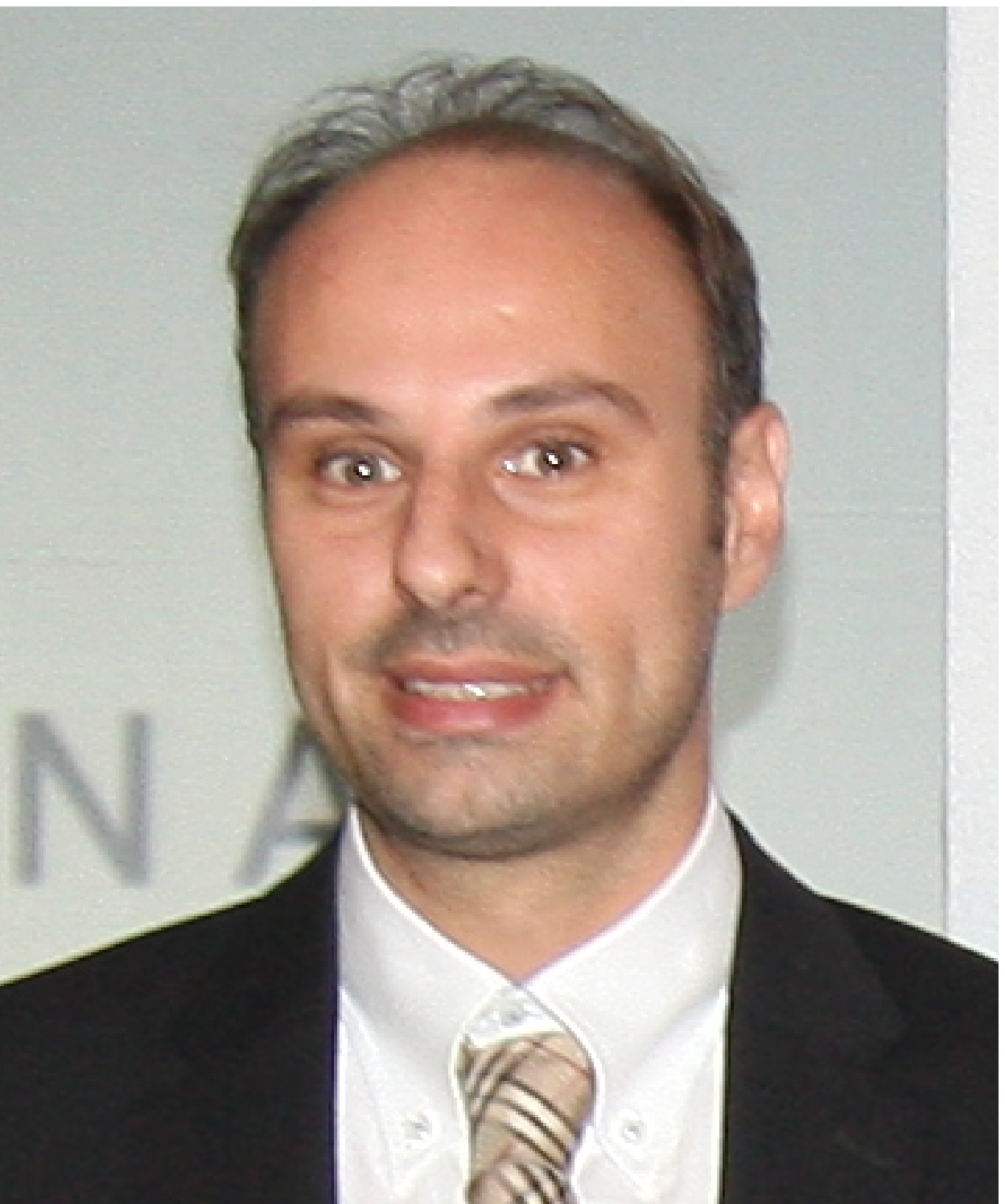}}]{Ioannis Krikidis} (S'03-M'07-SM'12) received the diploma in Computer Engineering from the Computer Engineering and Informatics Department (CEID) of the University of Patras, Greece, in 2000, and the M.Sc and Ph.D degrees from Ecole Nationale Sup\'erieure des T\'el\'ecommunications (ENST), Paris, France, in 2001 and 2005, respectively, all in electrical engineering. From 2006 to 2007 he worked, as a Post-Doctoral researcher, with ENST, Paris, France, and from 2007 to 2010 he was a Research Fellow in the School of Engineering and Electronics at the University of Edinburgh, Edinburgh, UK. He has held also research positions at the Department of Electrical Engineering, University of Notre Dame; the Department of Electrical and Computer Engineering, University of Maryland; the Interdisciplinary Centre for Security, Reliability and Trust, University of Luxembourg; and the Department of Electrical and Electronic Engineering, Niigata University, Japan. He is currently an Assistant Professor at the Department of Electrical and Computer Engineering, University of Cyprus, Nicosia, Cyprus. His current research interests include  information theory, wireless communications, cooperative communications, cognitive radio and secrecy communications.

Dr. Krikidis serves as an Associate Editor for the IEEE Transactions on Communications,  IEEE Transactions on Vehicular Technology, IEEE Wireless Communications Letters, and Wiley Transactions on Emerging Telecommunications Technologies.  He was the Technical Program Co-Chair for the IEEE International Symposium on Signal Processing and Information Technology 2013. He received an IEEE Communications Letters and IEEE Wireless Communications Letters exemplary reviewer certificate in 2012. He was the recipient of the {\it Research Award Young Researcher} from the Research Promotion Foundation, Cyprus, in 2013.
\end{IEEEbiography}

\begin{IEEEbiography}[{\includegraphics[width=1in,height=1.25in,clip,keepaspectratio]{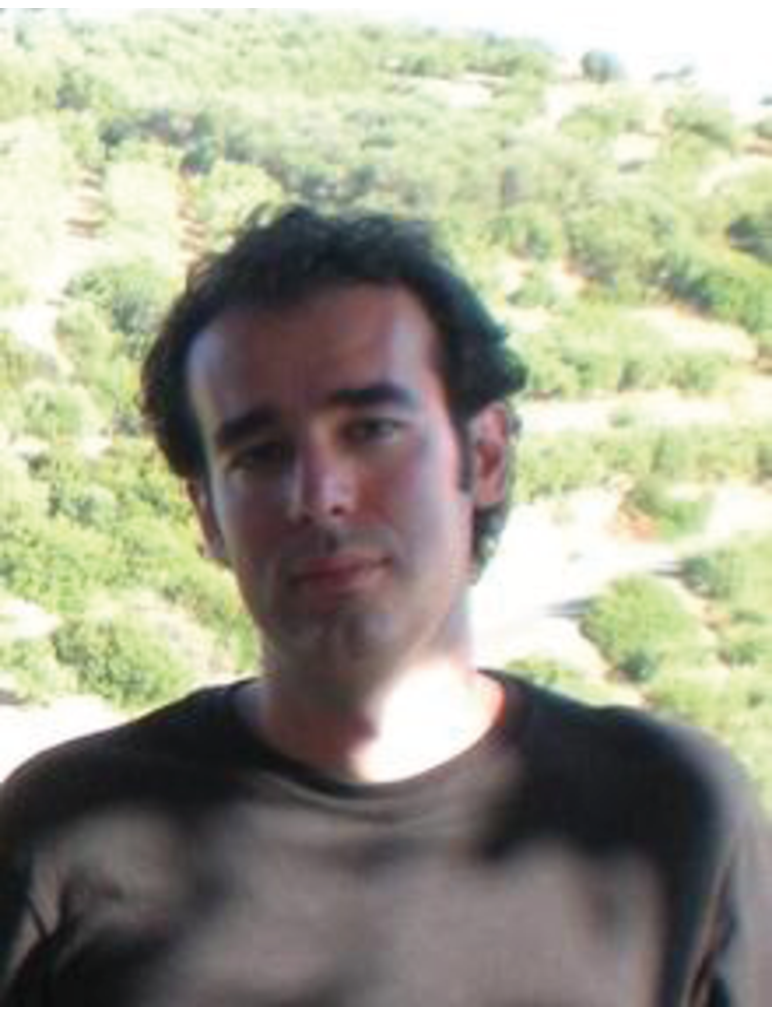}}]%
{Sotiris Karachontzitis} received his Diploma in computer engineering and informatics in 2004, and his M.Sc. degree in signal and image processing systems in 2008, both from the University of Patras, Greece. He is currently a Ph.D. student at the Computer Engineering and Informatics Department (CEID), University of Patras, Greece. His research interests mainly include resource allocation \& optimization in wireless communications, MIMO signal processing and secrecy communications. He is a student member of IEEE and a member of the Technical Chamber of Greece.
\end{IEEEbiography}

\begin{IEEEbiography}[{\includegraphics[width=1in,height=1.25in,clip,keepaspectratio]{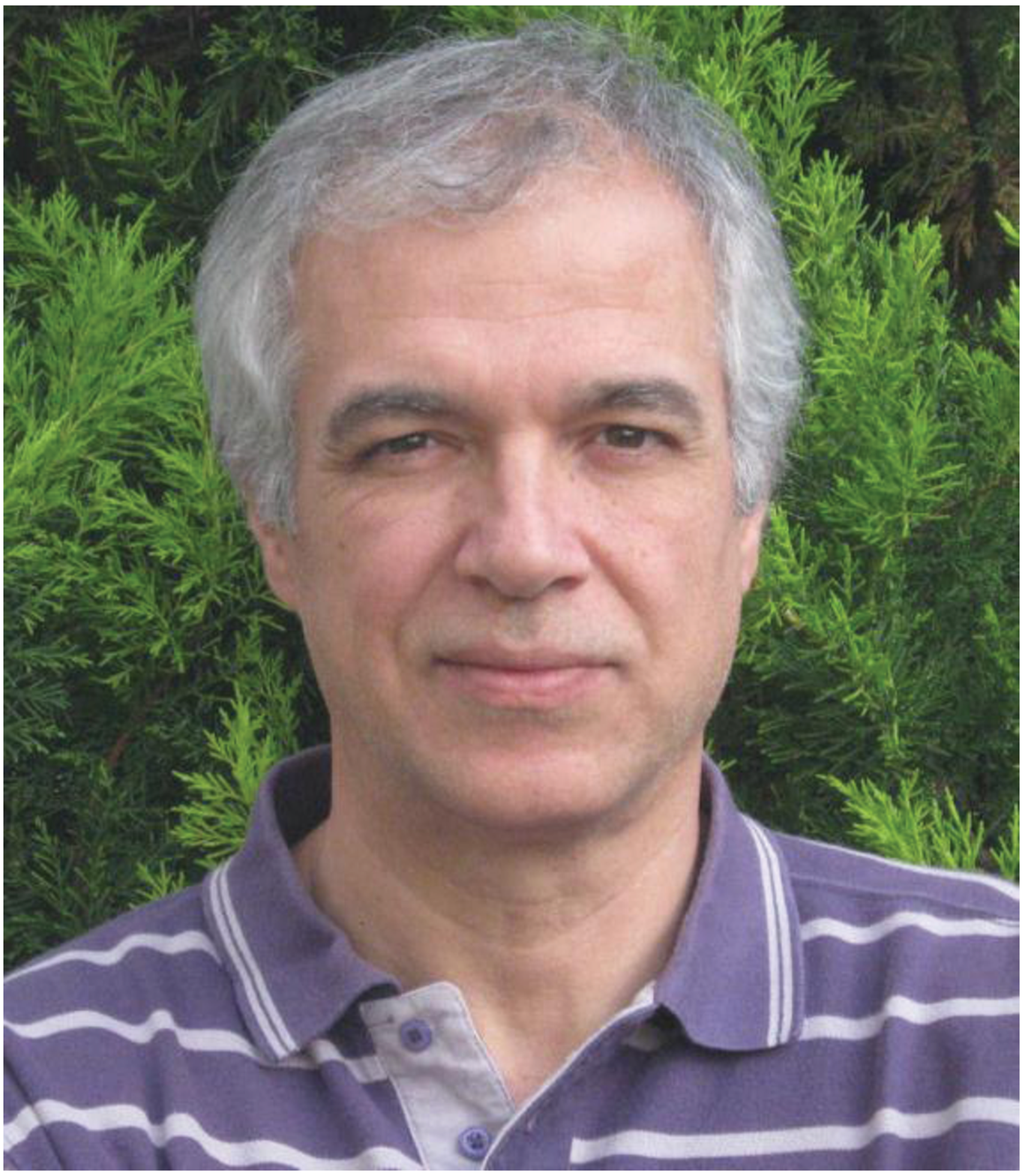}}]%
{Kostas Berberidis} (S'87-M'90-SM'07) received the Diploma degree in electrical engineering from DUTH, Greece, in 1985, and the Ph.D. degree in signal processing and communications from the University of Patras, Greece, in 1990. During 1991, he worked at the Signal Processing Laboratory of the National Defense Research Center. From 1992 to 1994 and from 1996 to 1997, he was a researcher at the Computer Technology Institute (CTI), Patras, Greece. In period 1994/95 he was a Postdoctoral Fellow at CCETT/CNET, Rennes, France. Since December 1997, he has been with the Computer Engineering and Informatics Department (CEID), University of Patras, where he is currently a Professor, and Head of the Signal Processing and Communications Laboratory.  Also, since 2008, he has been Director of the Signal Processing \& Communications Research Unit of the Computer Technology Institute and Press ``Diophantus''. His research interests include adaptive filtering, distributed processing, signal processing for communications, and wireless sensor networks. 

Prof. Berberidis has served or has been serving as a member of scientific and organizing committees of several international conferences, as Associate Editor for the IEEE Transactions on Signal Processing and the IEEE Signal Processing Letters, as a Guest Editor for the EURASIP JASP and as Associate Editor for the EURASIP Journal on Advances in Signal Processing. Also, he is a member of the Signal Processing Theory and Methods Technical Committee of the IEEE SPS and, since February 2010, he has been serving as Chair of the Greece Chapter of the IEEE Signal Processing Society. He is a Member of the Technical Chamber of Greece, a member of EURASIP, and a Senior Member of the IEEE. 
\end{IEEEbiography}

\end{document}